\documentclass[a4paper, 12pt]{article}
\usepackage{epsfig}
\usepackage{epic}
\usepackage{color}
\usepackage{graphicx}
\usepackage{setspace}
\usepackage{amssymb,amsmath,amsthm}
\usepackage{bbm}
\usepackage{subfigure}

\newtheorem{Theorem}{Theorem}

\newtheorem{Definition}{Definition}

\begin{document}

\title{Threshold Saturation in Spatially Coupled Constraint Satisfaction Problems}
\author{S.Hamed Hassani, Nicolas Macris and Ruediger Urbanke
\\
\\
\small{Laboratory for Communication Theory}\\
\small{School of Computer and Communication Science}\\
\small{Ecole Polytechnique F\'ed\'erale de Lausanne}\\
\small{Station 14, EPFL, CH-1015 Lausanne, Switzerland}}

\maketitle

\begin{abstract} 
\noindent 
We consider chains of random constraint satisfaction models that are spatially coupled across a finite window
along the chain direction. We investigate their phase diagram at zero temperature using the survey propagation formalism and the interpolation method. 
We prove that the SAT-UNSAT phase transition threshold of an infinite chain is identical to the one of the individual standard model,
and is therefore not affected by spatial coupling.
We compute the survey propagation complexity using population dynamics as well as large degree approximations, and determine the survey propagation 
threshold. We find that a clustering phase survives coupling. However,
as one increases the range of the coupling window,
the survey propagation threshold increases and saturates towards the phase transition threshold. 
We also briefly discuss other aspects of the problem.
Namely, the condensation threshold is not affected 
by coupling, but the dynamic threshold displays saturation towards the condensation one. 
All these 
features may provide a new avenue for obtaining better provable algorithmic lower bounds on phase transition thresholds of the
individual standard model.

\end{abstract}

\section{Introduction}\label{section 1}

The field of modern error correcting codes used on noisy communication channels has witnessed an interesting recent development.
Spatially coupled Low-Density Parity-Check (LDPC) codes, initially introduced\footnote{In its original form the construction goes under 
the name Terminated Convolutional Low-Density Parity-Check codes.}
by Felstroem and Zigangirov \cite{Felstrom-Zigangirov}, 
have been recognized to have excellent performance 
due to the {\it threshold saturation phenomenon} \cite{Kudekar-Richardson-Urbanke-I}. We refer to \cite{Kudekar-Richardson-Urbanke-II} 
for the history and review of contributions in the field of communications, and a general analysis of the phenomenon.

Recently we introduced an elementary statistical mechanical model, namely a chain of coupled Curie-Weiss spin systems \cite{ITW}, \cite{Curie-Weiss-Chain}, 
that already captures all the main features of this phenomenon, and shows that it is related to basic concepts of statistical mechanics. 
In a nutshell, this model is a one-dimensional chain of complete graph Ising models
coupled in the longitudinal direction\footnote{Similar models have already been considered in other contexts
\cite{Bak}, \cite{Lebowitz}.} by a Kac-like potential. When the range of the Kac potential goes to infinity (and its intensity to zero) 
the spinodal curve is pushed towards the coexistence one. The stable phase undergoes nucleation and grows, starting from the ends of the chain, 
and consequently (in the Kac limit) the metastability domain disappears from the phase diagram. This
has important algorithmic consequences in the context of error correcting codes, for the recovery of an original message from the 
corrupted one.

We already argued in \cite{ITW}, \cite{Curie-Weiss-Chain} that threshold saturation occurs quite generally when mean field models are coupled together into
a one-dimensional chain and the longitudinal range of the coupling increases to infinity. The individual mean field model may be some sort of spin glass system 
on a sparse Erdoes-Renyi like random graph or on a complete graph (or on an hyper-graph). The sparse case is relevant to error correcting codes; 
see \cite{Kudekar-Richardson-Urbanke-II} and references therein.
The case of complete graphs is relevant to compressive sensing, another topic to which these ideas have been successfully applied 
\cite{shrini}, \cite{mezardetal}, \cite{montanarietal}. 

Other models, defined on sparse random graphs, that are of great
interest both in theoretical computer science and statistical mechanics, are random Constraint Satisfaction Problems (CSP). 
We refer the reader who is not familiar with such problems to the recent book \cite{Mezard-Montanari-book}.
In this paper we 
investigate {\it random spatially coupled-CSP}. We specifically concentrate on two main representatives: satisfiability (SAT) and coloring (COL). 
We also briefly discuss the XOR-SAT problem which is somewhat similar
to the LDPC codes on the binary erasure channel, but has an interest of its own.

Here we focus on the {\it zero-energy} states of CSP. There are two ways to formulate the problem. One can directly
minimize the Hamiltonian, or one can study the uniform measure over zero-energy states. 
We will focus essentially on the first aspect. The second one will only be briefly discussed. 
We say that a SAT (resp. UNSAT) phase corresponds to 
a vanishing (resp. finite) average ground state energy per variable, in the thermodynamic limit. Since the ground state energy per variable concentrates,
this means that in a SAT (resp. UNSAT) phase
there are, with high probability, at most a sub-linear (resp. at least a linear)
number of unsatisfied constraints. 
In the language of computer science the problems that we are investigating are randomized versions of MAX-SAT and MAX-COL.

Coupled-CSP are based on chains of $L$ {\it individual} random bipartite graphs with 
constant degree $K$ for constraint nodes, and Poisson degree with mean $\alpha K$ for variable nodes. Each individual bipartite  
graph is appropriately {\it coupled} to its neighbors 
across a window of size $w$. The precise construction is explained in detail in section \ref{generalsetting}.
The number $\alpha > 0$ is a measure of the constraint density and plays the role of a control parameter. 
In the $K$-SAT problem each constraint 
corresponds to satisfying a disjunction of $K$ literals. In $Q$-COL, we have $K=2$ and all variable nodes 
connected to the same constraint node must have different colors in a $Q$-ary alphabet. 
Despite 
the similarities in construction with the LDPC case,
in general, CSP (and coupled-CSP) are considerably more difficult to analyze. To study the ground state
problem we adopt the {\it Survey 
Propagation (SP) formalism}, which is derived from the zero-temperature (level-$1$) cavity method of spin-glass theory \cite{cavity-level-1}. 
We refer to \cite{Mezard-Montanari-book} for a recent pedagogical account, but
for the convenience of the reader we review and adapt the formalism to coupled CSP, in a streamlined form, in appendix \ref{B}. 

Let us pause and explain the predictions of the SP formalism for individual graph ensembles \cite{mzp}, \cite{mezard-zecchina}. SP 
is a sophisticated mean field theory based on a set of
fixed point equations.
They predict the existence of a SAT-UNSAT phase transition when $\alpha$ crosses a critical 
threshold $\alpha_s$.  At a lower value $\alpha_{\rm SP}$ one finds a bifurcation from trivial to non-trivial solution of the fixed point equations. In 
the interval $[\alpha_{\rm SP}, \alpha_s]$ the solution space is fragmented in an exponentially 
large (in system size) number of well separated clusters of SAT ground states in Hamming space
(binary or $Q$-ary). 
The rate of growth of the number of such clusters with system size, is called the {\it zero-energy complexity} and is positive in the interval $[\alpha_{\rm SP}, \alpha_s]$. The complexity becomes formally negative at $\alpha_s$.

We consider the SP equations for coupled $K$-SAT and $Q$-coloring models and solve them by the method of population dynamics (sections \ref{example-k-sat} and \ref{example-q-col}). 
We find a positive complexity in an interval $[\alpha_{{\rm SP}, L, w}, \alpha_{s, L, w}]$ which allows to determine the SAT-UNSAT phase transition point $\alpha_{s, L, w}$ (where the complexity becomes formally negative). We make the following observations for the interval where the complexity is positive. 
We have that $\alpha_{s,L, w} > \alpha_s$ and $\alpha_{s,L,w}\downarrow \alpha_s$ as $L$ increases (and $w$ fixed). Interestingly we find that {\it threshold saturation takes place}, 
namely $\alpha_{{\rm SP}, L,w} \to \alpha_s$ as $L$ and $w$ both increase such that $1<<w<<L$. 
These findings are supported by a large $K$ and $Q$ analysis of the SP fixed point equations of coupled CSP
(sections \ref{example-k-sat} and \ref{example-q-col}). In this limit 
the fixed point equations reduce to one-dimensional difference equations, 
analogous to the ones found for the Curie-Weiss chain or coupled LDPC codes on the binary erasure channel. 
This allows to study an "average total warning probability" that characterizes the phase of the system. 
This quantity is somewhat analogous to the average magnetization in the CW chain, or the average erasure probability for LDPC codes. 
A corresponding ``van der Waals curve'' displays an oscillating structure around a ``Maxwell plateau''. 
Each oscillation corresponds to a state of the system characterized by a kink profile 
for a "local warning density" and a "local complexity density" along the chain. 

The thermodynamic limits of the average ground state energies  per node, for the chain and the individual ensembles are 
proven to be equal (section \ref{proofs}).
The proof uses an interpolation method \cite{guerra-toninelli-I}, \cite{guerra-toninelli-II}, \cite{franz-leone} in a convenient combinatorial form similar to \cite{BGP09}. 
This result is of some importance because it establishes that non-analytic points in the average ground state energy per node of the chain and individual ensembles, occur at the same constraint density. In other words one must have 
$\lim_{L\to +\infty}\alpha_{s,L,w} = \alpha_s$. 

In sections \ref{dyncond} and \ref{last-section} we briefly discuss further important aspects that will be the object of more detailed future work.

The SP formalism says nothing about the relative sizes (internal entropy) of clusters of solutions and does not take into account 
which of them are "relevant" to the uniform measure over zero energy solutions. This issue is addressed by the 
entropic cavity method \cite{entropiccavity}, \cite{lenka}, \cite{elaborate}, 
which allows to compute the so-called dynamical and condensation thresholds $\alpha_d$ and $\alpha_c$. We have computed 
the dynamical $\alpha_{d,L,w}$ and condensation $\alpha_{c,L,w}$ thresholds of coupled CSP, and observe that as $L$ increases 
$\lim_{L\to+\infty}\alpha_{c,L,w}\to \alpha_c$ ($w$ fixed) while $\alpha_{d, L,w}\to \alpha_c$ when both $w$ and $L$ increase in the regime
$1<<w<<L$. The first observation is consistent with a rigorous result
proved in appendix \ref{A}:  the thermodynamic limit of the free energy (at finite temperature) of the chain is identical 
to that of the individual model. These issues are discussed in section \ref{dyncond}. 

This work may have interesting algorithmic consequences.
Any bound on $\alpha_{s,L,w}$ can be turned into a bound for 
$\alpha_s$ by taking $L\to +\infty$ (note this is also true for the condensation threshold). 
In particular, algorithmic lower bounds on $\alpha_{s,L,w}$ can be turned into lower bounds for $\alpha_s$.
Now, because of the saturation of the SP and dynamical thresholds of coupled CSP, the values of $\alpha$ for which 
the space of zero-energy solutions is fragmented into well separated clusters, are substantially larger, compared to values of individual ensembles. Therefore one may 
hope that a form of {\it algorithmic threshold saturation}, or at least {\it algorithmic threshold increase}, happens for some 
well chosen algorithms applied to coupled CSP. This may allow to {\it prove} better algorithmic
lower bounds on $\alpha_{s,L,w}$ and thus $\alpha_s$. The proposed methodology is briefly discussed and 
illustrated in section \ref{last-section} with simple peeling algorithms.

\section{General Setting}\label{generalsetting}

We define a general class of CSP that form the {\it individual ensemble}. Then we couple these, to form  
one-dimensional chains called {\it spatially coupled-CSP ensembles}.

\subsection{Individual CSP ensemble $[N, K, \alpha]$.}\label{ind}

First, we specify an ensemble $(N, K, \alpha)$ of random bipartite graphs. 
Let $G=(V\cup C, E)$ with {\it variable} nodes $i\in V$, 
{\it constraint} nodes $c\in C$ and edges $\langle c, i\rangle$ connecting sets $C$ and $V$.
We have $\vert V\vert = N$, $\vert C\vert= M$, where $M= \lfloor \alpha N \rfloor$ (the integer part of $\alpha M$) and  $\alpha$ is a fixed number called the
constraint density. We call $N$ the size of the graph which is to be thought as large, $N\to +\infty$. 
All constraints $c$ have degree $K$, and each edge $\langle c, i\rangle$ emanating from $c$  is independently connected uniformly at random (u.a.r.)
to a node in $i\in V$. As $N\to +\infty$, the degrees of the variable nodes tend to independent identically distributed (i.i.d.) 
with distribution $\mathrm{Poisson}(\alpha K)$.

We denote by $\partial i$ the set of constraints connected to variable node $i$ and by $\partial c$ 
the set of variable nodes connected to a constraint $c$. 

For each graph $G$ of the ensemble $[N, K, \alpha]$ we define a Hamiltonian (or cost function). 
To the variable nodes $i\in V$ we attach variables $x_i\in \mathcal{X}$ taking values in a discrete 
alphabet $\mathcal{X}$. To each constraint $c\in C$ we associate a function $\psi_c(x_{\partial c})$ which depends only on the variables $x_{\partial c}=(x_i)_{i\in \partial c}$
connected to $c$. 
For constraint satisfaction problems  $\psi_c(x_{\partial c})\in\{0,1\}$; we say that the constraint is {\it satisfied} if $\psi_c(x_{\partial c})=1$
and {\it not satisfied} if $\psi_c(x_{\partial c})=0$. 
The total Hamiltonian is
\begin{equation}\label{hamiltonian-uncoupled}
 \mathcal{H}(\underline x) = \sum_{c\in C} (1-\psi_c(x_{\partial c})).
\end{equation}
For many problems of interest the functions $\psi_c$ are themselves random. This will be made precise in each specific example; the only important
condition is that the functions $\psi_c$ are i.i.d. for all $c\in C$.
The ground state energy is $\min_{\underline x}\mathcal{H}(\underline x)$, the minimum possible number of unsatisfiable constraints. Our 
main interest is in the average ground state energy per node
\begin{equation} \label{e_N}
 e_N(\alpha) = \frac{1}{N}\mathbb{E}[\min_{\underline x}\mathcal{H}(\underline x)]
\end{equation}
where the expectation is taken over the $[N, K, \alpha]$ ensemble and possibly over the randomness of $\psi_c$.

\subsection{Coupled-CSP ensemble $[N, K, \alpha, w, L]$.}\label{cou}
This ensemble represents a chain of coupled underlying ensembles. Figure \ref{draw} is a visual aid but gives only a partial view.
We align positions 
$z\in \mathbb{Z}$. On each position $z \in \mathbb{Z}$, we lay down $N$ variable nodes labeled $(i,z)\in V_z$, $i=1,\cdots, N$. We also lay down 
 $M= \lfloor \alpha N \rfloor$
check nodes labeled $(c,z)\in C_z$, $c=1,\cdots, M$. When the node labels are used as subscripts, say as in $a_{(i,z)}$ or $a_{(c,z)}$, we will simplify the 
notation to $a_{iz}$ or $a_{cz}$. 
Let us now specify how the set of edges, $E$, is chosen. Each constraint 
$(c,z)$ has degree $K$, in other words $K$ edges emanate from it. Each of these $K$ edges is connected  to variable nodes as 
follows: we first pick a position $z+k$ with $k$ uniformly random in 
the window $\{0,\cdots, w-1\}$, then we pick a node $(i,z+k)$ u.a.r. in $V_{z+k}$, and finally we connect $(c,z)$ to $(i,z+k)$. 
The set of edges emanating from $(i,z)$ can be decomposed as a union  $\cup_{k=0}^{w-1}\{\langle (c,z-k),  (i,z)\rangle \mid c\in C_z\}$.
Asymptotically as $N\to +\infty$, its cardinality is $\mathrm{Poisson}(\alpha K)$; and the cardinalities of each set in the union are
i.i.d. $\mathrm{Poisson}(\frac{\alpha K}{w})$.

Finally, take $L$ an even integer. Restrict the set of constraint nodes to $\cup_{z=-\frac{L}{2}+1,\cdots, +\frac{L}{2}} C_z$ and delete edges 
emanating from constraints
that do not belong to this set. Restrict the set of variable nodes 
to $\cup_{z=-\frac{L}{2}+1,\cdots,\frac{L}{2}+w-1} V_z$. 

\begin{figure}[htp]
\begin{centering}
\input{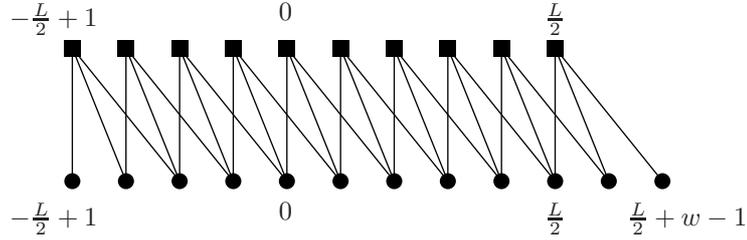}
\caption{{\small A representation of the geometry of the graphs with window size $w=3$ along
the ``longitudinal chain direction'' $z$. The ``transverse direction'' is viewed from the top. At each position there 
is a stack of $N$ variable nodes (circles) and a stack $M$ constraint nodes (squares). 
The depicted links between constraint and variable nodes represent stacks of edges.}}
\label{draw}
\end{centering}
\end{figure}

As in subsection \ref{ind}, we have a set of variables $x_{iz}\in \mathcal{X}$ and constraint functions $\psi_{cz}(x_{\partial cz})$ taking values
in $\{0,1\}$.
To each coupled graph in the ensemble we associate the Hamiltonian 
depending on $\underline x = (x_{iz})$, $(i,z)\in \cup_{z=-\frac{L}{2}+1,\cdots,\frac{L}{2}+w-1} V_z$,
\begin{equation}\label{hamiltonian}
 \mathcal{H}_{\rm cou}(\underline x) = \sum_{z=-\frac{L}{2}+1}^{\frac{L}{2}}\sum_{c\in C_z} (1-\psi_{cz}(x_{\partial(cz)})).
\end{equation}
The minimum over $\underline x$ is the ground state energy and its ensemble average per node is
\begin{equation}
 e_{N,L,w}(\alpha) = \frac{1}{NL}\mathbb{E}[\min_{\underline x}\mathcal{H}_{\rm cou}(\underline x)],
\end{equation}
where $\mathbb{E}$ is over the $[N, K, \alpha, w, L]$ graph ensemble and on the randomness in $\psi_{cz}$.

\vskip 0.25cm

\noindent{\bf Remark about the boundary conditions.} In the formulation above we have free boundary conditions. However, the average degree of 
the variable nodes close to the boundaries is reduced so that the CSP is easier close to the boundaries. Variable nodes close to the right boundary 
$z= \frac{L}{2}+1,...,\frac{L}{2}+w-1$ have degrees $\mathrm{Poisson}(\frac{\alpha K}{w}(\frac{L}{2} + w - z))$, and those close to the left boundary
$z= -\frac{L}{2}+1,...,-\frac{L}{2}+w-1$  have degrees $\mathrm{Poisson}(\frac{\alpha K}{w}(z+ \frac{L}{2}))$. It is 
sometimes convenient to imagine that the boundary nodes 
are connected to ``satisfied extra constraint nodes'', and all have $\mathrm{Poisson}(\alpha K)$ degree.

\subsection{K-Satisfiability and Q-Coloring.}\label{ksatqcol}
We define the main examples of constraint satisfaction problems that we analyze in this paper. 

{\it The $\text{K-SAT}$ problem.} The individual system is defined as follows. We take $x_i\in \{0,1\}$ the Boolean alphabet. 
Set $n(x_i)\equiv\bar x_i$ 
for the negation operation, and define $n^d(x_i) \equiv x_i$ when $d=0$ and $n^d(x_i)\equiv n(x_i)=\bar x_i$ when $d=1$. Pick 
$\text{Bernoulli}(\frac{1}{2})$ i.i.d. numbers $d_{\langle c, i\rangle}$ for each edge $\langle c, i\rangle\in E$. We say that an edge is {\it dashed} when 
$d_{\langle c, i\rangle} =1$ and {\it full} when $d_{\langle c, i\rangle} =0$. 
With this convention, a variable in a constraint is negated when it is connected to a dashed edge, and is not negated when it is connected to a full edge.
We set
\begin{equation}
 \psi_c(x_{\partial c}) = \mathbbm{1}(\vee_{i\in \partial c}( n^{d_{\langle c, i\rangle}}(x_i)) = 1).
\end{equation} 
These definitions are extended to the coupled system in an obvious way
\begin{equation}
 \psi_{cz}(x_{\partial (cz)}) = \mathbbm{1}(\vee_{iu\in \partial (cz)}( n^{d_{\langle cz, iu\rangle}}(x_{iu})) = 1),
\end{equation}
where the important point is that $d_{\langle cz, iu\rangle}$ are i.i.d. $\text{Bernoulli}(\frac{1}{2})$ for all edges.
The ground state energy counts the minimum possible number of unsatisfiable constraints. The instance is satisfiable iff the ground state energy is equal to zero.

\vskip 0.25cm
{\it The $\text{Q-COL}$ problem.} For the individual ensemble, we take $x_i\in \mathcal{X}=\{0,\cdots,Q-1\}$ the $Q$-ary color 
alphabet, $K=2$ for the constraint node degrees, and 
\begin{equation}\label{col}
 \psi_c(x_{\partial c}) = \mathbbm{1}(x_i\ne x_j \text{~for~} \{i,j\} = \partial c).
\end{equation}
Since the constraints have degree $2$ one can replace them by edges connecting directly $i$ and $j$ for $i,j\in \partial c$. The induced graph is, in the large size limit,
equivalent to the Erdoes-R\'enyi 
random graph $G(N, \frac{2\alpha}{N})$ ). The constraint
\eqref{col} forbids two neighboring nodes to have the same color. 

These definitions are easily extended to the 
coupled system. The induced graph (obtained by replacing constraints by edges) is now a coupled chain of Erdoes-R\'enyi graphs. In 
place of \eqref{col} we take $x_{iz}\in \mathcal{X}=\{0,\cdots,Q-1\}$ and 
\begin{equation}\label{colcoup}
 \psi_{cz}(x_{\partial (cz)}) = \mathbbm{1}(x_{iu}\ne x_{jv} \text{~for~} \{(i,u),(j,v)\} = \partial (c,z)).
\end{equation}
Given an instance of the induced graph, the ground state energy counts the minimum possible number of edges with vertices of the same color. 
The graph is colorable 
iff this number is zero.
\vskip 0.25cm

{\it The $\text{K-XORSAT}$ problem.} We briefly give relevant definitions that will be used in section \ref{last-section} and appendix \ref{A}. For the individual system 
$x_i\in \{0,1\}$ and $\psi_c(x_{\partial c}) = \mathbbm{1}(\oplus_{i\in\partial c}x_i=b_c)$ with $b_c$ being i.i.d. Bernoulli$(\frac{1}{2})$. Similarly for the coupled
system $\psi_{cz}(x_{\partial (cz)}) = \mathbbm{1}(\oplus_{iu\in\partial (cz)}x_{iu}=b_{cz})$ with $b_{cz}$ being i.i.d. Bernoulli$(\frac{1}{2})$.

\subsection{Static phase transition threshold}\label{static-phase}

For the purpose of analysis, it is 
useful to also consider an ensemble of coupled graphs with {\it periodic} boundary conditions. This ensemble is simply obtained from the $[N, K, \alpha, w, L]$ ensemble
by identifying the variable nodes $(i,z)$ at positions $z=\frac{L}{2} +k$ with nodes $(i,z)$ at positions 
$z=-\frac{L}{2}+k$ for each $k=1,\cdots, w-1$. The formal expression of the Hamiltonian 
$\mathcal{H}_{\rm cou}^{\rm per}(\underline x)$ is the same as in \eqref{hamiltonian} except that now 
$\underline x=(x_{iz})$ with $\cup_{z=-\frac{L}{2}+1,\cdots,\frac{L}{2}} V_z$. Quantities pertaining to this ensemble 
will be denoted by a superscript ''per``.


\begin{Theorem}[Comparison of open and periodic chains]\label{therm-limit-1}
For the general coupled-CSP $[N,K,\alpha, w,L]$ ensembles we have
\begin{equation}
 \vert e_{N, L,w}(\alpha) - e_{N,L,w}^{\rm per}(\alpha)\vert \leq \frac{\alpha w}{L}.
\end{equation}
\end{Theorem}
This theorem has an easy proof given in section \ref{proofs}.  

The next theorem does not have a trivial proof and is stated here for the special cases 
of $K$-SAT and $Q$-COL. While it is presumably valid for many other CSP's, we do not expect that it should hold in
complete generality. For example it is well known from models in other areas of statistical and condensed matter physics that the ground states of 
a periodic chain can break translation invariance (e.g crystals develop non-trivial periodic patterns) and then have lower energy than 
the homogeneous ground state. If this happens the statement of the theorem cannot possibly hold.
\begin{Theorem}[Thermodynamic limit]\label{therm-limit-2}
For the $K$-SAT and $Q$-COL models the two limits $\lim_{N\to+\infty} e_N(\alpha)$ and $\lim_{N\to +\infty}e_{N,L,w}^{\rm per}(\alpha)$ 
exist, are continuous, and non-decreasing in $\alpha$. Moreover they are equal,
\begin{equation}\label{equality-two-limits}
\lim_{L\to +\infty}\lim_{N\to +\infty}e_{N,L,w}^{\rm per}(\alpha) = \lim_{N\to+\infty} e_N(\alpha).
\end{equation}
\end{Theorem}
\vskip 0.25cm 
\noindent{\bf Remark about XORSAT.} We prove such a theorem for $K$-XORSAT {\it with $K$ even} in appendix \ref{A}. The proof breaks down for $K$ odd,
although the result is presumably true in that case also.
\vskip 0.25cm
Standard methods of statistical mechanics \cite{Ruelle} do not allow to prove the existence of the limits because the underlying graphs have expansion properties.
When the system is
cut in two parts the number of edges in the cut is of the same order as the size of the two parts and is not 
just a "surface" term. Therefore sub-additivity of the free and ground state energies become non-trivial.
However, interpolation methods allow to deal with this issue. The existence of the limit
for $\lim_{N\to+\infty} e_N(\alpha)$, as well as the fact that the function is continuous and non-decreasing, is proved for a range of models including the present ones 
in \cite{BGP09}, and it is easy to see that the same sort of proof works for the periodic chain. This proof will not be repeated. 
In section \ref{proofs} we provide the proof for the {\it equality} of the two limits. This is again based on two interpolations which provide upper and lower bounds.
Note that concentration of the ground state and free energies is also implied by standard arguments not 
discussed here\footnote{However concentration of the number of solutions in the SAT phase is more subtle
see \cite{abbe-montanari}.}.

We are interested in the thermodynamic limit 
$$
\lim_{\rm therm} \equiv \lim_{L\to +\infty}\lim_{N\to +\infty}
$$ 
for the open chain, which captures the regime of a long one-dimensional coupled-CSP. From theorems \ref{therm-limit-1} and \ref{therm-limit-2} we deduce that
\begin{equation}
\lim_{\rm therm}e_{N,L,w}(\alpha) = \lim_{\rm therm}e_{N,L,w}^{\rm per}(\alpha) = \lim_{N\to+\infty} e_N(\alpha).
\label{equal-energies}
\end{equation}
Since the energy functions are non-decreasing we can define a natural ``static phase transition'' threshold as follows.
\begin{Definition}[Static phase transition threshold]
 We define 
 \begin{equation}\label{prestatic}
 \alpha_{s,L, w} = \sup\{\alpha\vert \lim_{ N\to +\infty}e_{N,L,w}(\alpha) = 0\}
          =  \sup\{\alpha\vert \lim_{ N\to+\infty}e_{N,L,w}^{\rm per}(\alpha) = 0\},  
 \end{equation}
 and 
\begin{align}
 \alpha_s &= \sup\{\alpha\vert \lim_{\rm therm}e_{N,L,w}(\alpha) = 0\}
          =  \sup\{\alpha\vert \lim_{\rm therm}e_{N,L,w}^{\rm per}(\alpha) = 0\}  \nonumber \\ &
          =  \sup\{\alpha\vert \lim_{ N\to+\infty}e_{N}(\alpha) = 0\}.
\label{static}
\end{align}
\end{Definition}
The supremums in the first definition are equal because of theorem \ref{therm-limit-2} and those in the second definition 
are equal because of \eqref{equal-energies}. 
Note also that $\lim_{L\to +\infty}\alpha_{s,L,w} =\alpha_s$. The definition of $\alpha_s$ implies that,
for a given instance, when $\alpha<\alpha_s$ (resp. $\alpha>\alpha_s$) the number of unsatisfied constraints is $o(N)$ (resp. $O(N)$) with high probability. However 
it is not known how to automatically conclude that
a fixed instance is SAT (resp. UNSAT) with high probability when $\alpha<\alpha_s$ (resp. $\alpha>\alpha_s$). 

\vskip 0.25cm
\noindent{\bf Remark about finite temperatures.}
The theorems of this subsection have finite temperature analogs presented in appendix \ref{A}. As explained in section \ref{dyncond} these suggest that 
the condensation threshold obeys $\lim_{L\to +\infty}\alpha_{c,L,w} =\alpha_c$.
 
\subsection{Zero temperature cavity method and survey propagation formalism}\label{energetic}

We briefly summarize the simplest form of the cavity method and survey propagation equations for the coupled-CSP. More details on the formalism are presented in appendix \ref{B}.
When the graph instance is a {\it tree}, the minimization of \eqref{hamiltonian} can be carried out exactly. This leads to an expression 
for $\min_{\underline x} \mathcal{H}_{\rm cou}(\underline x)$ in terms of 
energy-cost messages $E_{iu\to cz}(x_{iu})$ and $\hat E_{cz\to iu}(x_{iu})$
that satisfy the standard min-sum equations (see equ. \eqref{warn1} and \eqref{warn2}). These messages are normalized so that 
$\min_{x_{iu}} E_{iu\to cz}(x_{iu}) = \min_{x_{iu}} \hat E_{cz\to iu}(x_{iu}) =0$ and they take values in $\{0, 1\}$. They may be interpreted as warning
messages. Roughly speaking, nodes inform each other on the most favorable values 
that the variable $x_{iu}$ should take in order to avoid energy costs. 
The ground state energy (on the tree) is given by the Bethe energy functional 
$\mathcal{E}[\{E_{iu\to cz}(.), \hat E_{cz\to iu}(.)\}]$ (see equ. \eqref{bethe}). 
For a {\it general graph instance} one considers the Bethe energy functional \eqref{bethe} as an ``effective Hamiltonian'' and studies the statistical mechanics
of this effective system. 
The min-sum equations are the stationary point equations of this functional
and the set of solutions  
$\{E_{iu\to cz}(.), \hat E_{cz\to iu}(.)\}$ characterize the  {\it state} of the system. 

It turns out  that the min-sum equations may have exponentially many (in system size) solutions 
with infinitesimal Bethe energy per node as $N\to +\infty$. A solution 
$\{E_{iu\to cz}^{(p)}(.), E_{cz\to iu}^{(p)}(.)\}$ with infinitesimal Bethe energy defines 
a {\it pure Bethe state}\footnote{We adopt this terminology to make a distinction with 
the mathematically precise notion of {\it pure state} for usual Ising models \cite{Ruelle}.} denoted by the superscript $(p)$. 
We define the average {\it zero-energy complexity} as
\begin{equation}
\Sigma_{L, w}(\alpha) = \lim_{\epsilon \to 0}\lim_{N\to +\infty}\frac{1}{NL}
\mathbb{E}[\ln ({\rm number~of~states~} p {\rm ~with~}\frac{\mathcal{E}^{(p)}}{N}=\epsilon)].
\label{complexity-definition}
\end{equation}
This quantity counts the number of pure Bethe states.
The typical behavior of the complexity as a function of $\alpha$ is as follows. 
Below an {\it SP threshold} it vanishes, then jumps to a positive value and decreases until it becomes negative at the static phase transition threshold. 
It therefore allows to compute
\begin{equation}\label{alphasp}
\alpha_{{\rm SP},L,w} = \inf\{\alpha\vert \Sigma_{L,w}(\alpha) > 0\},
\end{equation}
\begin{equation}\label{alphas}
\alpha_{s, L, w} = \sup\{\alpha\vert \Sigma_{L,w}(\alpha) > 0\}.
\end{equation}
It is a feature of the cavity theory that the static phase transition thresholds defined according to the energy \eqref{prestatic} and complexity
\eqref{alphas} coincide.  

The complexity is the Boltzmann entropy (on the zero energy shell) of the effective 
statistical mechanical problem with Hamiltonian \linebreak $\mathcal{E}[\{E_{iu\to cz}(.), E_{cz\to iu}(.)\}]$. 
It turns out that this can be computed, thanks to an effective partition function on the same sparse graph instance, again
within a message passing formalism. In this context messages are called {\it surveys}.
They count the fraction of pure Bethe states with given warning messages. Surveys $Q_{iu\to cz}(E_{iu\to cz}(.))$ and 
$\hat Q_{cz\to iu}(\hat E_{cz\to iu}(.))$
are exchanged between variable and constraint nodes according to survey propagation equations (see \eqref{sp1} and \eqref{sp2}).
The average complexity \eqref{complexity-definition} can be computed by a Bethe type formula for the entropy of the effective model.

The survey propagation equations \eqref{sp1}, \eqref{sp2} allow to compute 
 the distribution over pure Bethe states, of the vectors $(\hat E_{cz\to iu}(x_{iu}), x_{iu}\in \mathcal{X})$. 
These are $\vert\mathcal{X}\vert$-component vectors with components in $\{0, 1\}$. Thus the surveys are supported on 
an alphabet of size at most $2^{\vert\mathcal{X}\vert}$. Often
the effective size of the alphabet is smaller  (it is $\vert\mathcal{X}\vert+1$ in the specific problems considered here)
because the warning propagation equations \eqref{warn1}, \eqref{warn2} restrict the possible values 
of $(\hat E_{cz\to iu}(x_{iu}), x_{iu}\in \mathcal{X})$. This simplification is used for each model separately in the next sections. 

Let us summarize the main observations that follow from 
the detailed analysis in sections \ref{example-k-sat} and 
\ref{example-q-col}. As $L\to +\infty$, we find that the 
complexity curves $\Sigma_{L,w}(\alpha)$ supported on the interval
$[\alpha_{{\rm SP}, L, w}, \alpha_{s, L, w}]$ converge to a limiting curve 
$\Sigma_w(\alpha)$ supported on the limiting interval $[\alpha_{{\rm SP}, w}, \alpha_s]$. Moreover, on this later interval, $\Sigma_w(\alpha)$ coincides 
with the complexity $\Sigma(\alpha)$ of the individual system ($L=w=1$).
This is illustrated on Figure~\ref{complexi}. We observe that $\alpha_{s,L,w}$ tends to $\alpha_s$ from above. Also for moderate $L$ one generally has 
$\alpha_{{\rm SP}, L, w}>\alpha_s$, but this inequality is reversed for $L$ large enough, and $\lim_{L\to +\infty}\alpha_{{\rm SP}, L, w} = \alpha_{{\rm SP}, w}<\alpha_s$. 

We observe the threshold saturation, namely 
$\lim_{w\to +\infty}\alpha_{{\rm SP}, w} \uparrow \alpha_s$.
In fact we expect (from \cite{Curie-Weiss-Chain}) that the 
gap $\vert \alpha_{{\rm SP}, w} - \alpha_s\vert$ is exponentially small in $w$ 
($K$ fixed) but this is hard to assess numerically. One also observes that for 
$w$ fixed the gap increases with increasing $K$.

\begin{figure}[ht!]
\begin{centering}
\input{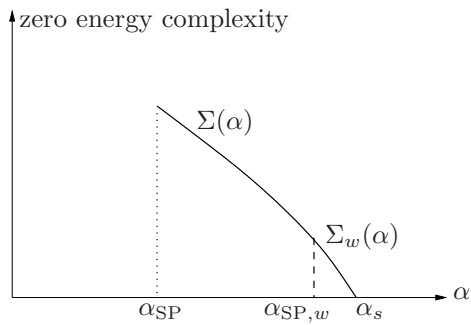}
\caption{\small Complexity of the individual ensemble $\Sigma(\alpha)$ (i.e. $L=w=1$) and limiting complexity $\Sigma_w(\alpha)$ of the coupled ensemble for 
$L\to +\infty$. We have $\alpha_{{\rm SP},w}\to\alpha_s$ as $w\to +\infty$.} 
\label{complexi}
\end{centering} 
\end{figure} 

We point out that the complexity of the chain with periodic boundary conditions 
converges to that of the individual system in the infinite length limit. In other words there is no threshold saturation  
as long as the boundary conditions are periodic. This is easily understood by realizing that the survey 
propagation equations are purely local and have a translation invariant solution when the 
boundary conditions are periodic.

Finally, let us mention that we observe similar features for the {\it entropic complexity} curve. 
In this case $\alpha_d$ plays the role of $\alpha_{SP}$ and $\alpha_c$ that of $\alpha_s$. We have $\alpha_{d,w}\to \alpha_c$. In particular,
$\lim_{L\to +\infty}\alpha_{c,L,w} =\alpha_c$ and \linebreak $\lim_{w\to +\infty}\lim_{L\to +\infty}\alpha_{d,L,w} =\alpha_c$ (see section \ref{dyncond}).

\section{Coupled $\text{K-SAT}$ Problem}\label{example-k-sat}

\subsection{Numerical implementation}\label{section4.1}

We begin with a convenient parametrization of the messages (see e.g \cite{Mezard-Montanari-book}).
Since $\mathcal{X}=\{0,1\}$, the warning (energy costs) messages are two-component 
vectors $(E_{iu\to cz}(0), E_{iu\to cz}(1))$
and $(\hat E_{cz\to iu}(0), \hat E_{cz\to iu}(1))$ which take three possible values 
$(0,1)$, $(1,0)$ and $(0,0)$. Warning $(0,1)$ means that $x_{iu}$ should take the value $0$, warning $(1,0)$ 
means that $x_{iu}$ should take value $1$, 
and warning $(0,0)$ means that $x_{iu}$ is free to take any value. Messages from variables to constraints can be conveniently parametrized as follows,
\begin{equation*}
Q_{iu\to cz}^S \equiv
\begin{cases}
Q_{iu\to cz}(0,1) & \mathrm{~if~} x_{iu} \mathrm{~is~negated~in~} cz ,\\
Q_{iu\to cz}(1,0) & \mathrm{if~} x_{iu} \mathrm{~is~not~negated~in~} cz.
\end{cases}
\end{equation*}
This is the fraction of pure states for which the variable is forced to satisfy the constraint. Similarly,
\begin{equation*}
Q_{iu\to cz}^U \equiv
\begin{cases}
Q_{iu\to cz}(0,1) & \mathrm{~if~} x_{iu} \mathrm{~is~not~negated~in~} cz, \\
Q_{iu\to cz}(1,0) & \mathrm{if~} x_{iu} \mathrm{~is~negated~in~} cz.
\end{cases}
\end{equation*}
This is the fraction of pure states for which the variable is forced to unsatisfy the constraint.
Note that $Q_{iu\to cz}(0,0) = 1- Q_{iu\to cz}^S - Q_{iu\to cz}^U$.
Let us now parametrize the messages from 
constraints to variables. If variable $x_{iu}$ enters unnegated in constraint $cz$, then certainly constraint $cz$ does not force 
it to take the value $0$. Thus $\hat Q_{cz\to iu}(0,1) =0$, and the message can be parametrized by the single number $\hat Q_{cz\to iu}(1,0)$. 
On the other hand, if variable $x_{iu}$ enters negated in constraint $cz$, then certainly constraint $cz$ does not force 
it to take the value $1$. Thus $\hat Q_{cz\to iu}(1,0) =0$, and again the message can be parametrized by the single number $\hat Q_{cz\to iu}(0,1)$. We set 
\begin{equation*}
 \hat Q_{cz\to iu} \equiv
\begin{cases}
\hat Q_{cz\to iu}(0,1) & \mathrm{~if~} x_{iu} \mathrm{~is~negated~in~} cz, \\
\hat Q_{cz\to iu}(1,0) & \mathrm{~if~} x_{iu} \mathrm{~is~not~negated~in~} cz.
\end{cases}
\end{equation*}
Message $\hat Q_{cz\to iu}$ is the fraction of pure states for which $cz$ warns $iu$ to satisfy it.
The survey propagation equations \eqref{sp1}, \eqref{sp2} then become
(recall $d_{\langle bv,iu\rangle} =1$ (resp. $0$) for a dashed (resp. full) edge 
$\langle bv,iu\rangle$),
\begin{equation}\label{hatq=produ}
 \hat Q_{cz\to iu} = \prod_{jv\in \partial(cz)\setminus iu} Q^U_{jv\to cz},
\end{equation}
and 
\begin{align}
Q_{iu\to cz}^S & \cong &  \biggl\{\prod_{bv\in \partial(iu)\setminus cz}^{d_{\langle bv, iu\rangle}\ne d_{\langle iu, cz\rangle}} (1- \hat Q_{bv\to iu})\biggr\}
\biggl\{1- \prod_{bv\in \partial(iu)\setminus cz}^{d_{\langle bv, iu\rangle}= d_{\langle iu, cz\rangle}} (1- \hat Q_{bv\to iu})\biggr\} ,
\label{qs}
\\
Q_{iu\to cz}^U & \cong & \biggl\{\prod_{bv\in \partial(iu)\setminus cz}^{d_{\langle bv, iu\rangle}= d_{\langle iu, cz\rangle}}  (1- \hat Q_{bv\to iu})\biggr\}
\biggl\{1- \prod_{bv\in \partial(iu)\setminus cz}^{d_{\langle bv, iu\rangle}\ne d_{\langle iu, cz\rangle}}  (1- \hat Q_{bv\to iu})\biggr\},
\label{qu}
\end{align}
where $\cong$ means that the r.h.s has to be normalized to one.
Define
\begin{align}
Q^+_{iu\to cz}  = &\prod_{bv\in \partial(iu)\setminus cz}^{d_{\langle bv, iu\rangle}= d_{\langle iu, cz\rangle}}  (1- \hat Q_{bv\to iu}) \label{qplus},\\
Q^-_{iu\to cz}  = &\prod_{bv\in \partial(iu)\setminus cz}^{d_{\langle bv, iu\rangle}\ne d_{\langle iu, cz\rangle}}  (1- \hat Q_{bv\to iu}).\label{qminus}
\end{align}
Then using \eqref{hatq=produ} and the normalized form of \eqref{qu}
\begin{equation}\label{hatq}
 \hat Q_{cz\to iu} = \prod_{jv\in \partial(cz)\setminus iu} \frac{Q^+_{jv\to cz} (1- Q^-_{jv\to cz})}{Q^+_{jv\to cz} + Q^-_{jv\to cz} - 
Q^+_{jv\to cz}Q^-_{jv\to cz}}.
\end{equation}
We will work with the set of SP equations \eqref{qplus}, \eqref{qminus}, \eqref{hatq}.
The complexity becomes
\begin{equation}
\Sigma_{L,w}(\alpha) = \frac{1}{NL}\mathbb{E}\biggl[\sum_{cz} \Sigma_{cz} +
\sum_{iz} \Sigma_{iz}  - \sum_{\langle cz, iu\rangle}\Sigma_{cz, iu}\biggr],
\end{equation}
with
\begin{equation}
\Sigma_{cz} = \ln\biggl\{\prod_{iu\in \partial(cz)} (Q_{iu\to cz}^+ +
Q_{iu\to cz}^-  -  Q_{iu\to cz}^+
Q_{iu\to cz}^- ) - \prod_{iu\in \partial(cz)}Q_{iu\to cz}^+(1-
Q_{iu\to cz}^-) \biggr\},
\end{equation}
\begin{equation}
\Sigma_{iz} = \ln\biggl\{\prod_{bv\in \partial(iz)}^{d_{\langle bv, iz\rangle}=1}(1-\hat Q_{bv\to iz}) +
\prod_{bv\in \partial(iz)}^{d_{\langle bv, iz\rangle}=0}(1-\hat Q_{bv\to iz}) - \prod_{bv\in \partial(iz)}(1-\hat Q_{bv\to iz}) \biggr\},
\end{equation}
\begin{equation}
\Sigma_{cz,iu} = \ln\biggl\{
(Q_{iu\to cz}^+ +
Q_{iu\to cz}^-  -  Q_{iu\to cz}^+
Q_{iu\to cz}^- ) 
-
Q_{iu\to cz}^+(1- Q_{iu\to cz}^-) \hat Q_{cz\to iu}.
\biggr\}
\end{equation}
The set of SP equations \eqref{qplus}, \eqref{qminus}, \eqref{hatq} is solved under the following assumptions.
We treat the set of messages emanating
from a constraint at position $z$, namely $\hat Q_{cz\to iu}$ for $u= z,\dots, z+w-1$, as i.i.d. 
copies of a r.v. $\hat Q_{z}$
depending only on the position $z$. Similarly we treat the messages emanating from a variable node at position $u$, namely 
 $Q^{\pm}_{iu\to cz}$ for $z=u-w+1, \dots, u,$
as i.i.d. copies of a r.v. $Q^{\pm }_{u}$. Now, fix a position $z$ and pick $p$, $q$ two independent Poisson$(\frac{\alpha K}{2})$ integers. 
Pick $k_1,\dots,k_{p+q}$ independently
uniformly in $\{0,\dots,w-1\}$. Similarly, pick $l_1,\dots,l_{K-1}$ independently uniformly in $\{0,\dots,w-1\}$. 
Under our assumptions the SP equations become\footnote{In \eqref{qplusrv}, \eqref{qminusrv}, \eqref{qhatrv} equalities mean that the r.v. have the same distribution.}
\begin{align}
Q^+_{z} & = \prod_{i=1}^p (1-\hat Q_{z-k_i}^{(i)}), \label{qplusrv}\\
Q^-_{z} & = \prod_{i={p+1}}^{p+q} (1-\hat Q_{z-k_i}^{(i)}),\label{qminusrv}
\end{align}
and 
\begin{equation}
 \hat Q_z = \prod_{i=1}^{K-1} \frac{Q^{+(i)}_{z+l_i}(1-Q^{-(i)}_{z+l_i})}{Q^{+(i)}_{z+l_i} + Q^{-(i)}_{z+l_i} - Q^{+(i)}_{z+l_i}Q^{-(i)}_{z+l_i}}.
 \label{qhatrv}
\end{equation}
The boundary conditions can be taken into account by setting $\hat Q_z=0$ for $z\leq -\frac{L}{2}$, $z>\frac{L}{2}$.
These equations are solved by the standard method of population dynamics. It is then possible to compute the average complexity from 
\begin{equation}\label{sigm}
\Sigma_{L,w}(\alpha) = \frac{1}{L}
\sum_{z=-\frac{L}{2}+1}^{\frac{L}{2}}(\alpha\mathbb{E}[\Sigma_{z}^{\rm cons}] + 
 \mathbb{E}[\Sigma_{z}^{\rm var}] - \alpha K\mathbb{E}[\Sigma_{z}^{\rm edge}]),
\end{equation}
where
\begin{equation}\label{sigmacons}
\Sigma_{z}^{\rm cons} = \ln\biggl\{\prod_{i=1}^K (Q_{z+l_i}^{+(i)} +
Q_{z+l_i}^{-(i)}  -  Q_{z+l_i}^{+(i)}
Q_{z+l_i}^{-(i)} ) - \prod_{i=1}^K Q_{z+l_i}^{+(i)}(1-
Q_{z+l_i}^{-(i)}) \biggr\},
\end{equation}
\begin{equation}\label{sigmavar}
\Sigma_{z}^{\rm var} = \ln\bigl\{\prod_{i}^{p}(1-\hat Q_{z-k_i}^{(i)}) +
\prod_{i=p+1}^{p+q}(1-\hat Q_{z-k_i}^{(i)}) - \prod_{i=1}^{p+q}(1-\hat Q_{z-k_i}^{(i)}) \bigr\},
\end{equation}
\begin{equation}\label{sigmaedge}
 \Sigma_{z}^{\rm edge} = \ln \biggl\{
(Q^+_{z+k} + Q^-_{z+k} - Q^+_{z+k}Q^-_{z+k}) - Q^+_{z+k}(1-Q^-_{z+k})\hat Q_{z}
\biggr\}.
\end{equation}
\vskip 0.25cm
\begin{figure}[htp]
\begin{centering}
\input{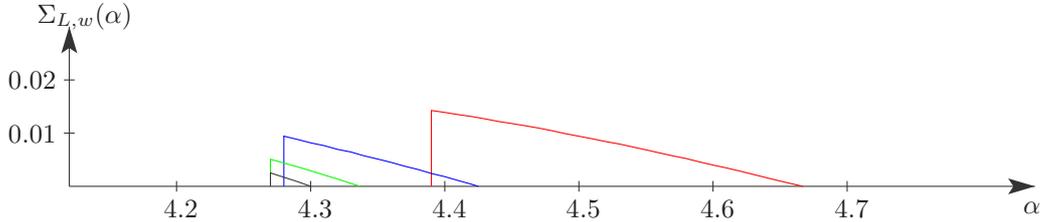}
\caption{\small Average complexity versus $\alpha$ for the $[1000,3,\alpha,3,L]$ ensembles with $L=10$ (rightmost curve), $20$, $40$, $80$ (leftmost curve).
Values of the corresponding thresholds are given in table \ref{3sattable}.}
\end{centering}
\label{3satcomplexity}
\end{figure} 
Figure \ref{3satcomplexity} shows the average complexity for the regime 
$N>>L>>w$, for $K=3$ and $w=3$.
We find it is positive in an interval $[\alpha_{SP, w, L}, \alpha_{s, w, L}]$
that shrinks down to zero as $L$ increases. The SP and phase transition
thresholds corresponding to Figure \ref{3satcomplexity} are given in Table 
\ref{3sattable}. We observe that $\alpha_{s,w,L}\downarrow \alpha_s$ as $L$ increases. Moreover we observe that {\it the SP threshold
saturates},
namely $\alpha_{{\rm SP}, w, L}\to \alpha_s$ for $L>>w>>1$. 

For moderate values of $L$ we have $\alpha_s < \alpha_{SP, w, L}$.
However since $\alpha_{SP, w, L} < \alpha_{s,w,L}$ and $\lim_{L\to +\infty} \alpha_{s,w,L} = \alpha_s$, for $L$ large enough and fixed $w$
we necessarily have $\alpha_{SP, w, L} < \alpha_{s}$. This turns out to be difficult to observe within population dynamics 
experiments, but can be checked in the large $K$ limit.
\begin{table}
\centering
\begin{tabular}{c c c c c c c c}
\hline\hline 
$\text{individual}$ & & $\alpha_{SP}$ & $\alpha_s$ \\
\hline
$L=1$ & &  $3.927$ &  $4.267 $ \\ 
\hline\hline
$\text{coupled}$ & & $\alpha_{{\rm SP},L,3}$ &   $\alpha_{s,L,3}$ \\
\hline
$L=10$                 & &  $4.386$ &  $4.663 $   \\
$L=20$                 & &  $4.274$ &  $4.425 $  \\
$L=40$                 & &  $4.269$ &  $4.335 $ \\
$L=80$                 & &  $4.268$ &  $4.301 $ \\
$L=160$               & &  $4.267$ &  $4.284 $ \\
\hline
\end{tabular}
\caption{\small SP and static phase transition thresholds of the 
$[1000,3,\alpha,3,L]$ ensembles.}
\label{3sattable}
\end{table} 



\subsection{Survey propagation for Large $K$}\label{section4.2}

For large $K$ one can derive approximations of the survey propagation equations that lend themselves to more 
explicit analysis \cite{mezard-mertens}. We will not attempt to control 
the error terms, but it is known for the individual system that the approximations are excellent already for $K\geq 5$. We can check numericaly 
that this is also the case for the coupled-CSP.
\vskip 0.25cm
\noindent{\bf Fixed point equations}.
Following  \cite{mezard-mertens}, we introduce entropic random variables
\begin{equation}\label{entropic}
 \hat q_z = -\ln (1-\hat Q_z), \qquad q_z^\pm = -\ln Q_z^{\pm}.
\end{equation}
From 
\eqref{qplusrv}, \eqref{qminusrv} and \eqref{qhatrv} we obtain 
\begin{equation}
q^{+}_{z} = \sum_{i=1}^{p} \hat q_{z-k_i}^{(i)},\quad
q^{-}_{z} =\sum_{i=p+1}^{p+q}\hat q_{z-k_i}^{(i)}, 
\label{qplusminus}
\end{equation}
and 
\begin{equation}\label{phi}
 \hat q_{z} = -\ln\biggl\{ 1- \prod_{i=1}^{K-1} \frac{e^{q_{z+l_i}^{-(i)}}-1}{e^{q_{z+l_i}^{-(i)}} + e^{q_{z+l_i}^{+(i)}} -1}\biggr\},
\end{equation}
we set
\begin{equation}\label{set}
\mathbb{E}[q^{\pm}_{z}] = x_z^\pm\quad {\rm and~~} \mathbb{E}[\hat q_{z}]= y_z,
\end{equation}
for the averages over the graph ensemble.
The number of i.i.d. random variables in \eqref{qplusminus} is a Poisson$(\frac{\alpha K}{2})$ integer. Therefore we assume 
that for large $K$ the r.v. $q^{\pm}_{z}$ are self-averaging.
It is reasonable to expect that they can be replaced by their expectation in \eqref{phi} and that hence $\hat q_{z}$ is also self-averaging.
This implies a closed set of equations for the expected values of messages,
\begin{equation}\label{second-average}
\begin{cases}
x_z^\pm  \approx \frac{\alpha K}{2w}\sum_{k=0}^{w-1} y_{z-k},
\\ 
y_z  \approx - \sum_{k_1,...,k_{K-1}=0}^{w-1} \frac{1}{w^{K-1}}
\ln\biggl\{ 1 - \prod_{i=1}^{K-1} \frac{e^{x^{-}_{z+k_i}} -1}{
e^{x^{-}_{z+k_i}} + e^{x^{+}_{z+k_i}} -1}\biggr\}.
\end{cases}
\end{equation}
We further approximate \eqref{second-average}. A self-consistent check with the final solution shows that $x^{\pm}=O(K)$ 
and hence
the product in the log is $O(2^{-K})$. 
Linearizing the logarithm yields
\begin{align}
y_z \approx
\sum_{k_1,...,k_{K-1}=0}^{w-1} \frac{1}{w^{K-1}} \prod_{i=1}^{K-1}
\frac{e^{x_{z+k_i}^{-}} -1}{
2e^{x_{z+k_i}^{-}} -1}
=
\biggl\{\frac{1}{w}\sum_{k=0}^{w-1}\frac{e^{x_{z+k}^{-}} -1}{
2e^{x_{z+k}^{-}}  -1}\biggr\}^{K-1}.
\end{align}
It is convenient to introduce the rescaled parameters
\begin{equation}\label{rescaled}
\hat\alpha=2^{-K}\alpha,\qquad \varphi_z = 2^{K-1} \hat\alpha
K y_z.
\end{equation}
From \eqref{entropic} we see $\varphi_z$ is a measure of the average (over the graph ensemble) probability (over pure states) that 
constraints at position $z$ send warning messages.
From now on we write $x_z$ instead of $x_z^{\pm}$. The fixed point equations become
\begin{equation}
 \begin{cases}
  x_z \approx \frac{1}{w}\sum_{k=0}^{w-1} \varphi_{z-k}, \\
  \varphi_z \approx \hat\alpha K \biggl\{\frac{1}{w}\sum_{l=0}^{w-1}\frac{e^{x_{z+l}}-1}
{e^{x_{z+l}}-\frac{1}{2}}\biggr\}^{K-1}.
\label{couple-of-equ}
 \end{cases}
\end{equation}
Hence, the profile $\{\varphi_z\}$ satisfies
\begin{equation}\label{coupledksat}
\varphi_z \approx \hat\alpha K \biggl\{\frac{1}{w}\sum_{k=0}^{w-1}\frac{e^{\frac{1}{w}\sum_{l=0}^{w-1}\varphi_{z-l+k}}-1}
{e^{\frac{1}{w}\sum_{l=0}^{w-1}\varphi_{z-l+k}}-\frac{1}{2}}\biggr\}^{K-1}.
\end{equation}
These equations have to be supplemented with the boundary condition $\varphi_z=0$ for $z\leq -\frac{L}{2}$ and $z>\frac{L}{2}$.
\vskip 0.25cm 
\noindent{\bf The average complexity.} Let us now express the complexity in terms of the fixed point profile. Let us first compute the contributions 
of variable and constraint nodes, and of edges.
\vskip 0.25cm
\noindent{\it Contribution of variable nodes}. From 
\eqref{sigmavar}, \eqref{entropic} and \eqref{set}
\begin{equation}
\Sigma_z^{var} = \ln\bigl\{e^{-\sum_{i=1}^p \hat q_{z-k_i}} +  e^{-\sum_{i=p+1}^q \hat q_{z-k_i}} - e^{-\sum_{i=1}^{p+q} \hat q_{z-k_i}}\bigr\}.
\end{equation}
For $K$ large the sums in the exponentials concentrate on their averages, so that
\begin{equation}
\mathbb{E}[\Sigma_z^{var}] \approx \ln
\bigl\{
2e^{-\frac{\alpha K}{2w}\sum_{k=0}^{w-1} y_{z-k}} - 
e^{-\frac{\alpha K}{w}\sum_{k=0}^{w-1}  y_{z-k}}
\bigr\}.
\end{equation}
\vskip 0.25cm
\noindent{\it Contribution of check nodes}.
From \eqref{sigmacons}, \eqref{entropic} and \eqref{set}
\begin{align}
\mathbb{E}[\Sigma_z^{cons} ]& = \mathbb{E}\biggl[\ln\biggl\{ 
\prod_{i=1}^K (e^{-q_{z+l_i}^{+}} +  e^{-q_{z+l_i}^{-}} - 
e^{-q_{z+l_i}^{+} - q_{z+l_i}^{-}})
\nonumber \\ &
-\prod_{i=1}^K e^{-q_{z+l_i}^{+}} (1- e^{-q_{z+l_i}^{-}})
\biggr\}\biggr]
\nonumber \\ &
\approx
\sum_{l_1,...,l_{K}=0}^{w-1}\frac{1}{w^K}
\ln\biggl\{ 
\prod_{i=1}^K (2e^{-x_{z+l_i}^{-}} - e^{-2x_{z+l_i}^{-}})
-\prod_{i=1}^K e^{-x_{z+l_i}^{-}} (1- e^{-x_{z+l_i}^{-}})
\biggr\}.
\end{align}
Factoring the first product out of the log we get
\begin{equation}
\mathbb{E}[\Sigma_z^{cons} ] \approx
\frac{K}{w}\sum_{i=0}^{w-1} \ln\bigl\{2 e^{-x^{-}_{z+l}} - e^{-2x^{-}_{z+l}}\bigr\}
 + 
\sum_{l_1,...,l_{K}=0}^{w-1}\frac{1}{w^K} 
\ln\biggl\{ 1 - \prod_{i=1}^{K} \frac{ 1- e^{-x_{z+l_i}^{-}} }{ 2 - e^{-x^{-}_{z+l}}  }
       \biggr\}.
\end{equation}
Since the ratio in the second log is $O(2^{-K})$ we can linearize and obtain 
\begin{equation}
\mathbb{E}[\Sigma_z^{cons} ]  \approx
\frac{K}{w}\sum_{l=0}^{w-1} \ln\bigl\{2 e^{-x^{-}_{z+l}} - e^{-2x^{-}_{z+l}}\bigr\}
 - 
 \biggl\{\frac{1}{w}\sum_{l=0}^{w-1}\frac{ 1- e^{-x_{z+l}^{-}} }{ 2 - e^{-x^{-}_{z+l} } }
 \biggr\}^K.
 \end{equation}
\vskip 0.25cm
\noindent{\it Contribution of edges}. Similarly from \eqref{sigmaedge}, \eqref{entropic}, \eqref{set} we have
\begin{align}
 \mathbb{E}[\Sigma_z^{\rm edge}] & = \frac{1}{w}\sum_{l=0}^{w-1}\mathbb{E}\biggl[\ln\bigl\{ ( e^{-q^+_{z+l}} + e^{-q^-_{z+l}} - e^{-q^+_{z+l}-q^-_{z+l}})
\nonumber \\ &
- e^{-q^+_{z+l}} (1- e^{-q^-_{z+l}}) (1- e^{-\hat q_z}) \bigr\}\biggr]
\nonumber \\ &
\approx\frac{1}{w}\sum_{l=0}^{w-1} 
\ln\biggl\{ ( 2e^{-x^-_{z+l}} - e^{-2x^-_{z+l}})
- e^{-x^-_{z+l}} (1- e^{-x^-_{z+l}}) (1- e^{-y_z}) \biggr\}.
\end{align}

Now, using \eqref{second-average} we can express the  total average complexity \eqref{sigm} in terms of rescaled variables \eqref{rescaled}. We find
\begin{equation}\label{avcomp}
\Sigma_{w,L}(\hat\alpha)  = \frac{1}{L}\sum_{z= -\frac{L}{2} +1}^{\frac{L}{2}} \sigma_{\hat\alpha, w,L}(z),
\end{equation}
with
\begin{align}
\sigma_{\hat\alpha, w,L}(z) & \approx
\ln
\bigl\{
2e^{-\sum_{k=0}^{w-1} \varphi_{z-k}} - 
e^{-\frac{2}{w}\sum_{k=0}^{w-1}  \varphi_{z-k}}
\bigr\}
-
2^K\hat\alpha\biggl\{  
\frac{1}{w}\sum_{l=0}^{w-1} 
\frac{ e^{x_{z+l}} - 1 }{ 2e^{x_{z+l}} - 1  }
\biggr\}^K
\nonumber \\ &
-
\frac{2^K\hat\alpha K}{w}\sum_{l=0}^{w-1} 
\ln\biggl\{ 1
- \frac{e^{x_{z+l}}-1}
{2e^{x_{z+l}}-1} (1- e^{-\frac{\varphi_z}{\hat\alpha K2^{K-1}}})\biggr\}.
\label{local-complexity}
\end{align}
Within our approximations the third term can be simplified further because  
$1- e^{-\frac{\varphi_z}{K2^{K-1}}} = O(2^{-K})$ and we may linearize the log. Thus the second 
line in \eqref{local-complexity} can 
be replaced by
\begin{equation}
  2\varphi_z\frac{1}{w} \sum_{l=0}^{w-1}\biggl\{\frac{e^{x_{z+l}}-1}
{2e^{x_{z+l}}-1}\biggr\}.
\end{equation}
The complexity \eqref{avcomp} can be viewed as a functional of the profiles $\{x_z, \varphi_z\}$ with boundary condition
$\varphi_z=0$ for $z\leq -\frac{L}{2}$ and $z>\frac{L}{2}$. One can check that the stationary points of this functional
are given by the fixed point equations \eqref{couple-of-equ}.

\subsection{Solutions for Large $K$}\label{largeKsolutions}
We use the notation $f\doteq g$ to mean that $\lim_{K\to +\infty}\frac{f}{g}=1$.                                         
The large $K$ results for the individual system \cite{mezard-mertens} are 
recovered by setting $L=w=1$, in which case the fixed point equations \eqref{coupledksat}
reduces to 
\begin{equation}
\varphi \approx \hat\alpha K \biggl\{\frac{e^\varphi -1}{e^\varphi -\frac{1}{2}}\biggr\}^{K-1}.
\label{fixed-k-sat1}
\end{equation}
One may easily check that this is the stationary point equation for 
the complexity \eqref{avcomp} as a function of $\varphi$ (and $\alpha$ fixed),
\begin{equation}
 \Sigma_{1,1}(\hat \alpha,\varphi) = \ln\{2e^{-\varphi} - e^{-2\varphi}\} 
- 2K\hat\alpha\biggl\{\frac{e^\varphi - 1}{2e^\varphi -1}\biggr\}^K 
+\varphi \biggl\{\frac{e^\varphi - 1}{2e^\varphi -1}\biggr\}.
\label{fixed-k-sat2}
\end{equation}
Thus, fixed points of \eqref{fixed-k-sat1} are stationary points of \eqref{fixed-k-sat2}: stable fixed points 
correspond to minima and unstable ones to maxima. 

The curve $\hat\alpha(\varphi)$ is shown as the dotted curve 
in Figure \ref{coupled-ksat-curves}. This function is convex and has a unique 
minimum 
at $\varphi_{\rm SP} \doteq \ln(\frac{1}{2}K\ln K)$ and $\hat\alpha(\varphi_{\rm SP})\equiv
\hat\alpha_{\rm SP} \doteq \frac{\ln K}{K}$. Near this minimum we have 
$\hat\alpha(\varphi)\approx (\frac{\varphi - \varphi_{\rm SP}}{\gamma_{\rm SP}})^2$,
$\gamma_{\rm SP} \doteq   \frac{4}{3} \frac{K}{\ln K}$. For $\varphi>>\varphi_{\rm SP}$ we 
have  $\hat\alpha(\varphi)=\frac{1}{K}(\varphi-\varphi_{\rm SP})$ and for 
$0<\varphi<<\varphi_{\rm SP}$ we have $\hat\alpha(\varphi) = \frac{1}{\varphi}$. Therefore
the trivial fixed point $\varphi=0$ is unique for
$\hat\alpha < \hat\alpha_{\rm SP}$, and there are two extra non-trivial fixed points for $\hat\alpha > \hat\alpha_{\rm SP}$. 
Only one of them is stable and forms the branch $\varphi_{\rm mst} \approx K\hat\alpha + \varphi_{\rm SP}$ for $\varphi>>\varphi_{\rm SP}$.

For $\hat\alpha<\hat\alpha_{SP}$, the function \eqref{fixed-k-sat2} has a unique minimum at $\varphi=0$. For $\hat\alpha>\hat\alpha_{SP}$ a second metastable minimum appears
at $\varphi_{\rm mst}\approx K\hat\alpha +\varphi_{\rm SP}$. At this minimum we find $\Sigma_{1,1}(\hat\alpha, \varphi_{\rm mst}) \doteq \ln 2 - \hat\alpha$ which counts the number of 
clusters as long as it is positive. Summarizing, the complexity vanishes 
for $\hat \alpha<\hat \alpha_{SP}$, and equals $(\ln 2 - \hat\alpha)$ for 
$\hat\alpha\in [\hat \alpha_{SP}, \ln2]$. 
In particular the static phase transition threshold is $\hat\alpha_{s}\doteq \ln2$.
Beyond the static phase transition threshold  
the complexity is negative and looses its meaning (one has to modify the SP formalism used here).
Higher order corrections can be computed 
in powers of $2^{-K}$, see \cite{mezard-mertens}.

Let us now discuss the coupled case. The picture which emerges is similar to the one for the much simpler 
Curie-Weiss Chain model \cite{Curie-Weiss-Chain} and coupled 
LDPC codes over the binary erasure channel \cite{Kudekar-Richardson-Urbanke-I}. Before discussing
the numerical results we wish to give a heuristic argument that ``explains'' why threshold saturation occurs.
The argument can presumably be turned into a rigorous proof using the methods in \cite{Kudekar-Richardson-Urbanke-I}
for LDPC codes on the binary erasure channel, or more general methods developed in \cite{SRT-one-dim}.

For the sake of the argument suppose that we fix $\hat\alpha>\hat\alpha_{\rm SP}$ and that we 
look for profile solutions of \eqref{coupledksat},
on an infinite chain $L\to +\infty$, that interpolate 
between the (asymmetric) boundary conditions $\varphi_z=0$, $z\to -\infty$ 
and $\varphi_z\to\varphi_{\rm mst}$, $z\to +\infty$. We take as an ansatz, a kink approaching 
its asymptotic values (at the two ends)
fast enough, with a transition region localized in a region of size $O(w)$ centered at a 
position $z_{\rm kink}=\xi L$ ($\vert\xi\vert\leq 1/2$). Figure~\ref{kink_profile} gives an illustrative picture of the kink profile.
We have
\begin{equation}
\overline\varphi\equiv\frac{1}{L}\sum_{z=-\frac{L}{2}+1}^{\frac{L}{2}}\varphi_z
\approx \frac{1}{L}(\frac{L}{2} - \xi L)\varphi_{\rm mst}.
\end{equation}
Also, it is easy to see that the associated complexity as a function of $\xi$, or equivalently $\overline\varphi$, is approximately given 
by a convex combination of the two
minima of $\Sigma_{1,1}(\alpha, \varphi)$ (given in \eqref{fixed-k-sat2}) which correspond 
to the two points $\varphi=0$ (with $\Sigma =0$) and $\varphi=\varphi_{\text{st}}$ (with $\Sigma \approx \ln 2- \hat{\alpha}$).
More precisely,
\begin{align*}
 \Sigma_{kink}(\xi) & \approx \frac{1}{L}\bigl [(\frac{L}{2} +\xi L) \times 0 +(\frac{L}{2} -\xi L) \times (\ln 2 -\hat\alpha) \bigr ] \\  
&\approx \frac{\overline\varphi}{\varphi_{\rm mst}} (\ln 2 -\hat\alpha). 
\end{align*}
\vskip 0.15cm
\begin{figure}[htp]
\begin{centering}
\input{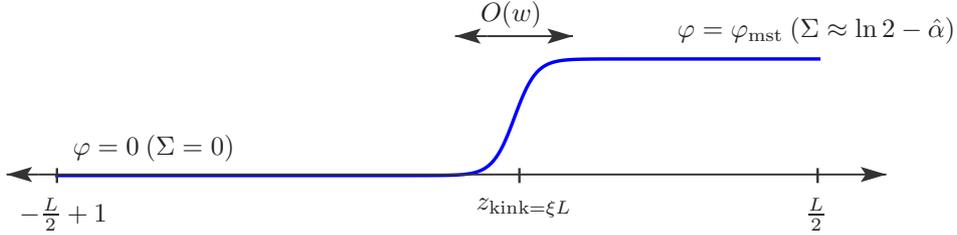}
\caption{{\small An illustrative picture of a kink-like ansatz $\{\varphi_z\}_{z=-\frac{L}{2}+1}^{\frac{L}{2}}$ for a solution of \eqref{coupledksat}. 
At the right end, the kink
converges to the value $\varphi=\varphi_{\text{st}}$ (with corresponding 
complexity $\Sigma \approx \ln 2- \hat{\alpha}$) and at the left end it converges to $\varphi=0$ (with $\Sigma =0$). 
The transition region of size $O(w)$ which is centered at $z=z_{\text{kink}}$.}}
\label{kink_profile}
\end{centering}
\end{figure}
\vskip 0.25cm
When $\hat\alpha<\hat\alpha_s$, the minimum is at $\xi=\frac{1}{2}$ ($\overline\varphi=0$). This means that the kink center will form a traveling
wave through the chain, and reach its unique stable location at the right end. On the other hand when $\hat\alpha>\hat\alpha_s$
the  minimum is at $\xi=-\frac{1}{2}$ ($\overline\varphi=\varphi_{\rm mst}$) and the kink will travel towards the left to reach its stable location.
Within the present approximation, for $\hat\alpha=\hat\alpha_s$ any position along the chain is stable for the kink center. 

Summarizing, this heuristic argument suggests that for $\hat\alpha <\hat\alpha_s$ the fixed point equations \eqref{coupledksat} only have the trivial 
solution $\{\varphi_z=0\}$, while for $\hat\alpha >\hat\alpha_s$ the only solution is
$\{\varphi_z=\varphi_{\rm mst}\}$. This means that the SP threshold coincides with $\hat\alpha_s$.
Here, $\xi$ has been treated as a continuous 
variable, which is expected to be valid only in a limit of large $w$. For large but finite $w$ there will subsist a small gap between 
the SP and static thresholds, and for $\hat\alpha$ fixed in this gap only a discrete set of positions for the kink are stable. The number of such stable positions is roughly
equal to $2L$.

We have solved \eqref{coupledksat}  numerically with {\it symmetric} boundary conditions $\varphi_z=0$, $z\leq -\frac{L}{2}$, $z>\frac{L}{2}$ and 
fixed $\overline\varphi \equiv \frac{1}{L}\sum_{z=-\frac{L}{2}}^{\frac{L}{2}}\varphi_z$. In order to find a solution for {\it all} values of $\overline\varphi$
we have to let $\hat\alpha$ vary slightly. In other words we find a solution  
$(\hat\alpha(\overline\varphi); \{\varphi_z(\overline\varphi)\})$ that is parametrized by $\overline\varphi$.
Define the {\it van der Waals curve} (Figure~\ref{coupled-ksat-curves}) as the function $\hat\alpha(\overline\varphi)$. The minimum of the van der Waals curve 
yields (as for the individual system) the SP threshold $\alpha_{{\rm SP}, w, L}$ (see Table \ref{tablelargeK} for numerical values).
\vskip 0.25cm
\begin{figure}[htp]
\begin{centering}
\input{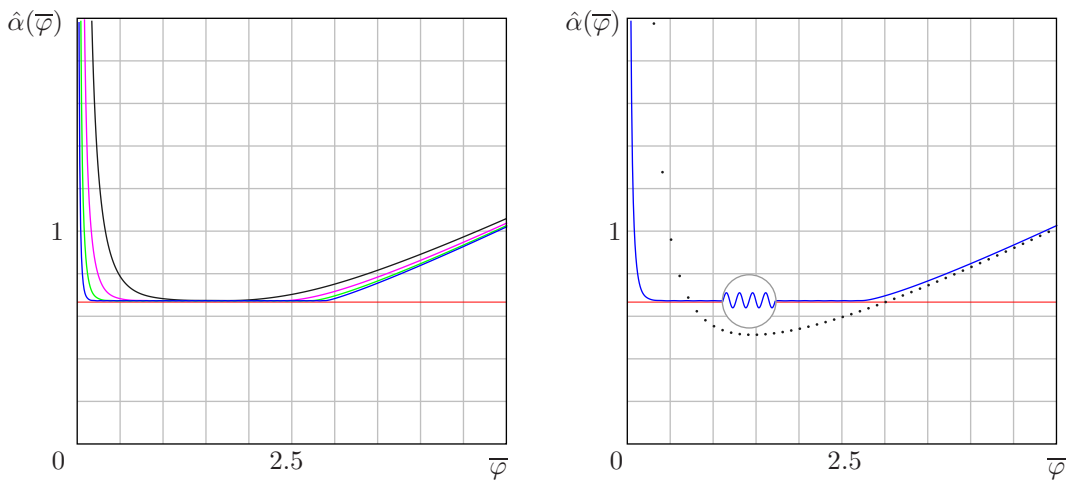}
\caption{{\small {\it Left}: sequence of van der Waals curves $\hat\alpha(\overline\varphi)$, for $K=5$, $w=3$ and $L= 10, 20, 40, 80$ (top to bottom). For
$\overline\varphi\in[\varphi_{\rm mst}, +\infty]$ they converge to the individual system curve. 
{\it Right}: a magnification of the plateau region for $K=5$, $w=3$ and $L=40$ shows the fine structure. The dotted line is 
the curve for the individual system and the red line shows the static phase transition threshold $\hat{\alpha}_s=0.666$.}}
\label{coupled-ksat-curves}
\end{centering}
\end{figure}
\vskip 0.25cm
As $L$ increases, the curves 
develop a plateau at height $\approx\hat\alpha_{s}$ for the interval $\overline\varphi\in [0, \varphi_{\rm mst}]$. 
Moreover they converge to the van der Waals curve of the
individual system for $\overline\varphi\in [\varphi_{\rm mst}, +\infty[$, a fact that is consistent with theorems \ref{therm-limit-1}, \ref{therm-limit-2}.
Precise enough numerics show that as long as $w$ is finite the curves display a fine structure in the plateau interval:
the magnification in Figure~\ref{coupled-ksat-curves} shows wiggles of very small amplitude. We observe that 
their amplitude decays as $w$ grows and $K$ is fixed (we expect from \cite{Curie-Weiss-Chain}
that this decay is exponential); and grows larger as $K$ increases with $w$ fixed (see Table~\ref{tablelargeK}).
\vskip 0.25cm
\begin{figure}[ht!]
\begin{centering}
\input{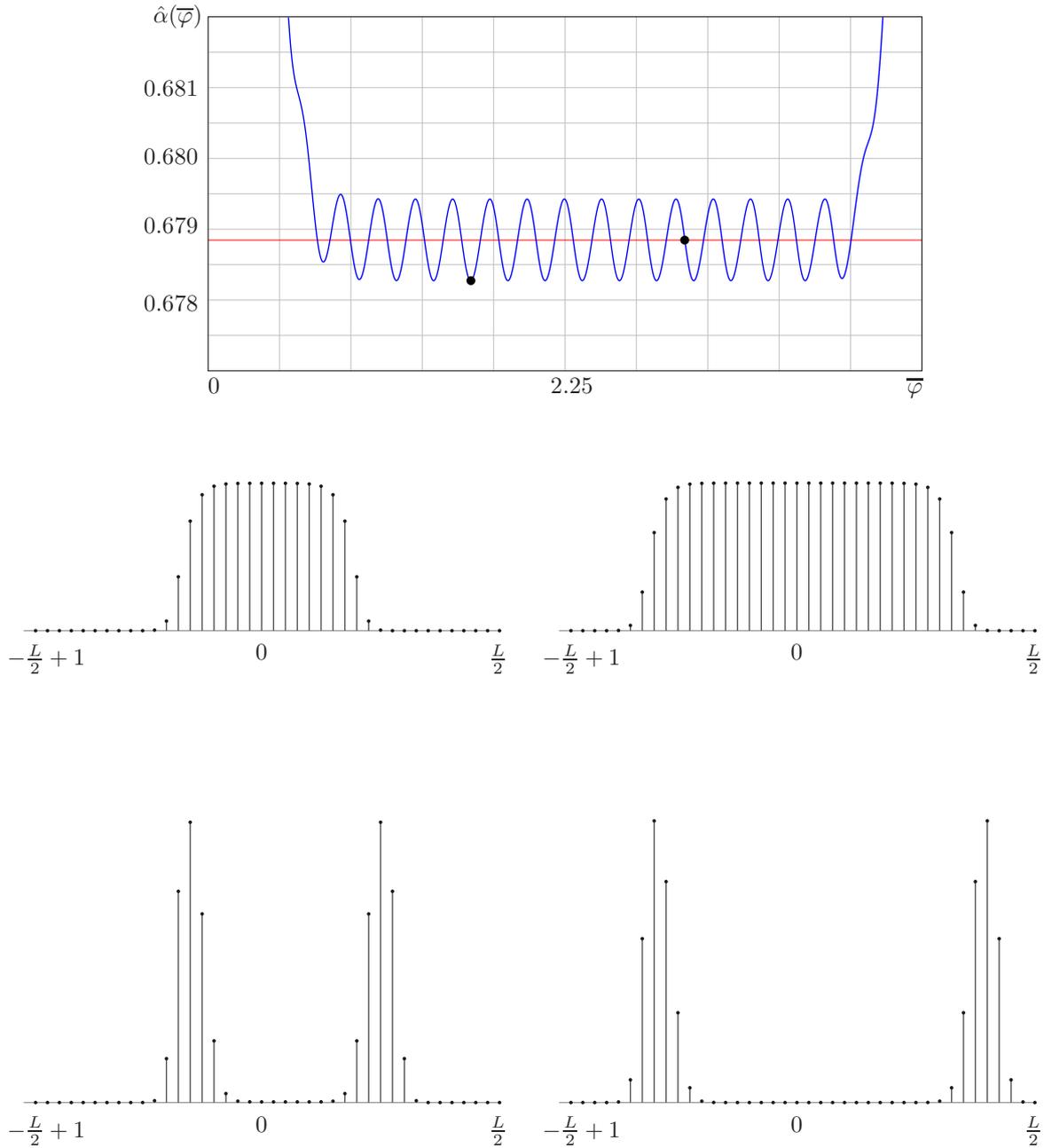}
\caption{{\small van der Waals curve in the wiggle region for the  coupled system (top) for $K=7$, $w=4$ and $L=40$. The red line 
is at the static phase transition threshold. The left point $(\overline{\varphi}_l, \hat{\alpha}_l)=(1.657, 0.678274)$ corresponds 
to the warning (middle left) and complexity (bottom left) profiles. In the latter the height of the middle part is
$\hat\alpha_s-\hat\alpha_l\approx 0.010$.  
The right point $(\overline{\varphi}_r, \hat{\alpha}_r)=(3.00585,0.688847)$ 
corresponds to the warning and complexity 
profiles 
on the right.}}
\label{compl-prof}
\end{centering}
\end{figure}
\begin{table}
\centering
\begin{tabular}{c c c c c c}
\hline\hline 
 K & $5$ & &   $7$  & &   $10$ \\
\hline
$\hat\alpha_{s}$ & $0.666$ & & $0.686$ & & $0.692$ \\
\hline
$\hat{\alpha}_{\text{SP}}$   &  $0.513$ & & $0.449$ & &  $0.370$ \\
\hline
$\hat{\alpha}_{\text{SP},80,3}$   &  $0.672$ & & $0.682 $ & &  $0.651$ \\
$\hat{\alpha}_{\text{SP},80,5}$   &  $0.672$ & & $0.688 $ & &  $0.691$ \\
$\hat{\alpha}_{\text{SP},80,7}$   &  $0.672$ & & $0.688 $ & &  $0.692$ \\
\hline
\end{tabular}
\caption{{\small SP Thresholds of the 
individual $(L=w=1)$ and coupled ensembles ($L= 80, w= 3, 5, 7)$ are found from the van der Waals curves. 
For $K=10$ we clearly 
see that the SP threshold saturates to $\hat\alpha_s$ from below as $w$ increases.}}
\label{tablelargeK}
\end{table} 
Figure \ref{compl-prof} illustrates the solutions of the fixed point equations for $\hat\alpha$ in the wiggle region for large $K$.
The top curve is the van der Waals curve in the wiggle region.
The middle left warning density profile is the fixed point solution corresponding the left point with coordinates
$(\overline\varphi_l, \hat\alpha_l)$. Note that $\hat\alpha_l = \hat\alpha_{SP, L, w}$. For this point the total average complexity is approximately equal to 
$\frac{\overline\varphi_l}{\varphi_{\rm mst}}(\hat\alpha_s - \hat\alpha_l)$. The bottom left curve shows the complexity profile. In the middle part, the 
height of this profile is approximately $(\hat\alpha_s - \hat\alpha_l)\approx 0.010$. Consider 
now the point on the right with coordinates $(\overline{\varphi}_r, \hat{\alpha}_r)$. Note that we take this point very close to
the static phase transition threshold $\hat\alpha_r\approx \hat\alpha_s$. As a consequence the total average complexity nearly vanishes.
The middle right warning density profile is flat over the whole chain,
except near the ends because we enforce the boundary conditions, and the complexity density nearly
vanishes except in the transition regions.

\section{Coupled Q-Coloring Problem}\label{example-q-col}

\subsection{Numerical implementation}

First we introduce an adequate parametrization of the messages (see e.g \cite{Mezard-Montanari-book}). The
warning vectors $(E_{jv\to cz}(1), \cdots , E_{jv\to cz}(Q))$ fall in two categories: those 
that have {\it exactly one} zero component; and those that have {\it at least two} zero components. For coloring, equation \eqref{warn2} becomes
\begin{equation}
 \hat E_{cz\to iu}(x_{iu}) = \min_{x_{\partial cz}\setminus iu}\bigl\{\mathbbm{1}(x_{cz}=x_{jv}) + E_{jv\to cz}(x_{jv})\bigr\}
 - \hat C_{cz\to iu}.
\end{equation}
It is easy to see that when $E_{jv\to cz}$ has exactly one zero component, then $\hat E_{cz\to iu}$ has exactly one non-zero component. 
On the other hand,
 when $E_{jv\to cz}$ has at least two zero components then all components of $\hat E_{cz\to iu}$ are zero. Hence, the vector
$(\hat E_{cz\to iu}(1),\cdots,\hat E_{cz\to iu}(Q))$ can take only $Q+1$ possible values which are the $(0,\cdots,0)\equiv *$ vector and 
the $Q$ canonical basis vectors $(1,0,\cdots,0)\equiv 1$, $(0,1,0,\cdots,0)\equiv 2$, ..., $(0,\cdots,1)\equiv Q$. The interpretation is clear: 
a warning vector $\in\{1,\cdots,Q\}$ forces the variable to choose a color, while a warning vector $*$ leaves the variable free.

We can rewrite the SP equations in terms of the distribution of warnings $\hat Q_{cz\to iu}(a)$,
$a\in\{1,\cdots, Q, *\}$.
Since constraints have degree $2$
we can view the messages $\hat Q_{cz\to iu}(a)$ as carried by the edge $\langle jv, iu\rangle$, where $jv$ is the unique 
node in $\partial(cz)\setminus iu$.  We thus make the replacement
$\hat Q_{cz\to iu}(a) \to \hat Q_{jv\to iu}(a)$ and write down the SP equations on the induced graph of variable nodes. 
Moreover following \cite{colorpeople} we seek solutions that do not depend on colors, and set 
$\hat Q_{jv\to iu}(a) \equiv \hat Q_{jv\to iu}$ for $a\in \{1,\dots, Q\}$. A calculation then shows that \eqref{sp1}, \eqref{sp2} reduce 
to
\begin{equation}
\hat Q_{jv\to iu} = \frac{\sum_{l=1}^Q (-1)^l {Q-1\choose{l}} \prod_{kw \in \partial(jv)\setminus iu} (1- (l+1)\hat Q_{kw\to jv})}
{\sum_{l=0}^{Q-1}(-1)^l {Q\choose{l+1}} \prod_{kw\in \partial(jv)\setminus iu} (1- (l+1)\hat Q_{kw\to jv})},
\label{spcol1}
\end{equation}
and  
\begin{equation}
 \hat Q_{jv\to iu}(*) = 1-Q \hat Q_{jv\to iu}.
\label{spcol2}
\end{equation}
Now, recall that the degrees of nodes of the induced graph are Poisson$(2\alpha)$ integers. {\it From now on we set} $c\equiv 2\alpha$. 
We solve \eqref{spcol1}, \eqref{spcol2}
under the assumption that the messages emanating from node $jv$ are i.i.d. copies of a random variable $\hat Q_v$ with a distribution
that depends only on the position $v$. Fix a position $z$,   
pick an integer $d$ according to a Poisson$(c)$, and pick integers $k_1,\cdots, k_d$ i.i.d. uniform in $\{-w+1,\cdots, w-1\}$. 
We have\footnote{Interpreted as equalities between random variables.}
\begin{equation}\label{qcoldistr}
\hat Q_{z} = \frac{\sum_{l=1}^Q (-1)^l {Q-1\choose{l}} \prod_{1=1}^d (1- (l+1)\hat Q_{z+k_i}^{(i)})}
{\sum_{l=0}^{Q-1}(-1)^l {Q\choose{l+1}} \prod_{i=1}^d (1- (l+1)\hat Q_{z+k_i}^{(i)})},
\end{equation}
and 
\begin{equation}
 \hat Q_{z}(*) = 1-Q \hat Q_{z}.
\end{equation}
Here, the relevant boundary conditions are taken into account by setting $\hat Q_z =0$ for $z\leq -\frac{L}{2}$ and $z>\frac{L}{2}$.
These equations are solved numerically by population dynamics. This then allows to compute the ensemble average of the 
complexity,
\begin{equation}\label{avcompcol}
\Sigma_{L,w}(c) = \frac{1}{L}\sum_{z=-\frac{L}{2}+1}^{\frac{L}{2}}(\mathbb{E}[\Sigma_z^{\rm var}]
 - \frac{c}{2}\mathbb{E}[\Sigma_z^{\rm edge}]),
\end{equation}
with 
\begin{align}
 \Sigma_z^{\rm var} & =  \ln\biggl\{\sum_{l=0}^{Q-1} (-1)^l{{Q}\choose{l+1}}\prod_{i=1}^{d}(1-(l+1)
\hat Q_{z+k_i}^{(i)})\biggr\} \label{vcol},\\
 \Sigma_z^{\rm edge}& = \ln\bigl\{1-Q \hat Q_z^{(1)}\hat Q_{z+k}^{(2)}\bigr\}\label{ccol}.
\end{align}
\begin{figure}[htp]
\begin{centering}
\input{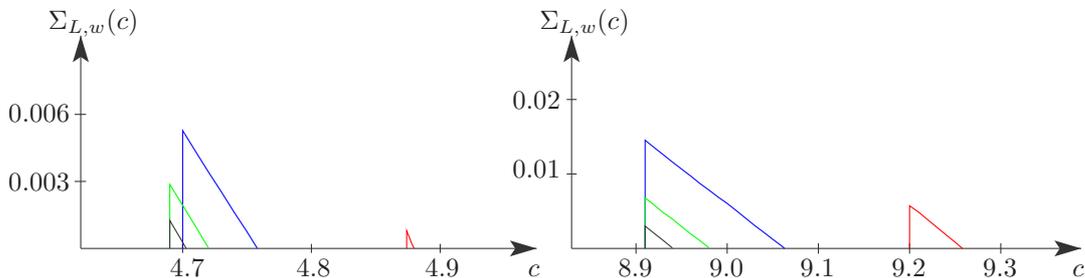}
\caption{\small Average complexity for the $[1000,Q,c,3,L]$ ensembles with $Q=3$ (left) and $Q=4$ (right). Here $L=10$, $20$, $40$, $80$ from right to left.
See corresponding table \ref{tablecoloring} for numerical values of thresholds.}
\label{qcol-comp}
\end{centering}
\end{figure}

The numerical results are similar to those for the coupled $K$-SAT model. Figure~\ref{qcol-comp}
shows that the complexity is positive in an interval $[c_{{\rm SP}, w, L}, c_{s, w, L}]$ which signals the existence of exponentially 
many pure states. Beyond $c_{s, w, L}$ the complexity becomes negative which means that the graph instances are not colorable w.h.p.
Table \ref{tablecoloring} gives the values of these thresholds. Again, we observe that $c_{s,w,L}\downarrow c_s$ as $L$ increases, and that threshold saturation takes place,
namely $c_{{\rm SP}, w, L}\to c_{s}$ as $L>>w>>1$.
\begin{table}
\centering
\begin{tabular}{c c c c c c c c c c}
\hline\hline 
$Q$                  & &   $3$                &  $3$              & &  $4$                &   $4$           \\
\hline
$\text{individual}$  & &  $c_{\rm SP}$        &  $c_s$            & &  $c_{\rm SP}$       &   $c_s$         \\
\hline
$L=1$                & &  $4.42$              &  $4.69 $          & &  $8.27$             &   $8.90$        \\ 
\hline\hline
$\text{coupled}$     & & $c_{{\rm SP}, L, 3}$ &  $c_{s, L, 3}$    & &  $c_{{\rm SP},L,3}$ &   $c_{s, L, 3}$ \\
\hline
$L=10$               & &  $4.874$              &  $4.879 $          & & $9.20$              &  $9.25$         \\
$L=20$               & &  $4.70$              &  $4.75 $          & & $8.91$              &  $9.06$         \\
$L=40$               & &  $4.69$              &  $4.72 $          & & $8.91$              &  $8.98$         \\
$L=80$               & &  $4.69$              &  $4.70 $          & & $8.91$              &  $8.93$         \\
\hline
\end{tabular}
\caption{{\small Thresholds computed by population dynamics for the individual and coupled ensembles for ensembles $[1000, Q, c, 3, L]$.}}
\label{tablecoloring}
\end{table} 

\subsection{Survey propagation for large $Q$}

\noindent{\bf Fixed point equations.} The large $Q$ analysis for the individual system \cite{colorpeople} is extended to the coupled model.
In this regime $c_s$ is $O(Q\ln Q)$ therefore it is natural to set
\begin{equation}
c= \hat c \, Q\ln Q,
\end{equation}
and to analyze \eqref{qcoldistr} for $\hat c$ fixed. In this limit the node degrees concentrate on $d\approx Q\ln Q$ with a fluctuation 
$O(\sqrt{Q\ln Q})$. 
Therefore, we assume that in the expression
\begin{equation}\label{productcol}
\prod_{i=1}^d (1- (l+1) \hat Q_{z+k_i}^{(i)}) = e^{\sum_{i=1}^d \ln\{1- (l+1) \hat Q_{z+k_i}^{(i)}\}},
\end{equation}
we can replace the sum over a large number of terms by its average,
\begin{equation}
e^{\frac{\hat c Q\ln Q}{2w-1}\sum_{k=-w+1}^{w-1}\mathbb{E} [\ln\{1 -(l+1) \hat Q_{z+k}\}]}. 
\end{equation}
In this expression the average over $d$ and $k_1,\dots, k_d$ has been carried out and the remaining expectation is over $\hat Q_z$.
Since the product \eqref{productcol} enters in \eqref{qcoldistr}, we conclude that $Q_z$ concentrates on its average.  Thus, setting
\begin{equation}
 \mathbb{E}[\hat Q_{z}] = \hat q_z
\end{equation}
for the average warning probability,
we find
\begin{equation}\label{intermediate}
 \hat q_z = \frac{\sum_{l=1}^Q (-1)^l {Q-1\choose{l}} \exp\bigl(\frac{\hat c Q\ln Q}{2w-1}\sum_{k=-w+1}^{w-1}\ln\{1 -(l+1) \hat q_{z+k}\}\bigr)}
{\sum_{l=0}^{Q-1}(-1)^l {Q\choose{l+1}} \exp\bigl(\frac{\hat c Q\ln Q}{2w-1}\sum_{k=-w+1}^{w-1}\ln\{1 -(l+1) \hat q_{z+k}\}\bigr)}.
\end{equation}
It can be checked self-consistently from the solutions of the fixed point equation that $\hat q_z=O(Q^{-1})$ and therefore for $l=O(1)$ the log in 
\eqref{intermediate} can be linearized, while the terms with higher $l$ are damped. Linearizing the log the sum over $l$ can be performed,   
and working with rescaled variables
\begin{equation}\label{rescaled-col}
 \theta_z\equiv (\hat c Q\ln Q) \hat q_z,
\end{equation}
we find
\begin{equation}\label{qcolfixedpoint}
 \theta_z = \hat c Q\ln Q
\frac{\sum_{l=1}^Q (-1)^l {Q-1\choose{l}} \exp\bigl(-\frac{l+1}{2w-1}\sum_{k=-w+1}^{w-1} \theta_{z+k})\bigr)}
{\sum_{l=0}^{Q-1}(-1)^l {Q\choose{l+1}} \exp\bigl(-\frac{l+1}{2w-1}\sum_{k=-w+1}^{w-1} \theta_{z+k})\bigr)}.
\end{equation}
Let
\begin{equation}
 F_Q(\theta) =Q\ln Q \,e^{-\theta} \frac{(1-e^{-\theta})^{q-1}}{1-(1-e^{-\theta})^{q}}.
\end{equation}
The fixed point equation takes the simple form 
\begin{equation}\label{one-dim-Q}
\theta_z = \hat c F_Q\bigl(\frac{1}{2w-1}\sum_{k=-w+1}^{w-1} \theta_{z+k}\bigr).
\end{equation}
These equations must be solved with the boundary condition  $\theta_z=0$ for $z\leq -\frac{L}{2}$ and $z>\frac{L}{2}$ in order
to find the average warning probability profiles.
\vskip 0.15cm
\noindent{\bf Average complexity.} Proceeding as above we find from \eqref{vcol}, \eqref{ccol} 
\begin{align}
 \mathbb{E}[\Sigma_z^{\rm var}]  & = \ln\biggl\{\sum_{l=0}^{q-1} (-1)^l {Q\choose l+1} \label{varcol}
e^{\frac{\hat c Q\ln Q}{2w-1}\sum_{k=-w+1}^{w-1} \ln\{1-(l+1)\hat q_{z+k}\}}\biggr\}, \\
 \mathbb{E}[\Sigma_z^{\rm cons}] & = \frac{1}{2w-1}\sum_{k=-w+1}^{w-1}\ln\bigl\{1-Q \hat{q}_z\hat{q}_{z+k}\bigr\}. \label{conscol}
\end{align}
As before since $\hat q_z= O(Q^{-1})$ we can linearize the log in the exponential of the first equation and the one in the second equation. 
Then working with the rescaled variables \eqref{rescaled},
straightforward algebra leads to an average complexity \eqref{avcompcol} given by
\begin{align}\label{coloringcomplexity}
\Sigma_{L,w}(c)  = \frac{1}{L}\sum_{z=-\frac{-L}{2}+1}^{\frac{L}{2}}\biggl[\ln&\bigl\{1- (1- e^{-\frac{1}{2w-1}\sum_{k=-w+1}^{w-1} \theta_{z+k}})^Q\bigr\}
\nonumber \\ &
+ \frac{\theta_z}{2\hat c\ln Q}\frac{1}{2w-1}\sum_{k=-w+1}^{w-1}\theta_{z+k}\biggr].
\end{align}
This functional is defined for profiles that satisfy the boundary condition $\theta_z=0$ for $z\leq -\frac{L}{2}$ and $z>\frac{L}{2}$.
The consistency of our approximations can be checked by noticing that the stationary points 
of \eqref{coloringcomplexity} are precisely given by the solutions of the fixed point equation \eqref{one-dim-Q}.

\subsection{Solutions for large Q}

The discussion is quite similar to the case of $K$-SAT so we will be brief. By setting $L=w=1$, we recover the fixed point equation of the individual system
which is $\theta = \hat c F_Q(\theta)$. Fixed points are stationary points of 
 the complexity as a function of $\theta$,
\begin{equation}
 \Sigma_{1,1}(\hat c, \theta) = \ln\bigl\{1-(1-e^{-\theta})^Q\bigr\} + \frac{\theta^2}{2\hat c\ln Q} 
\label{single-col-comp}
\end{equation}
This function controls the existence and nature of the fixed points. 
At $\theta=0$ it has 
a minimum for all $\hat c$ which corresponds to a trivial stable fixed point and a vanishing complexity
$\Sigma_{1,1}(\hat c)=0$. It is unique for 
$\hat c<\hat c_{\rm SP} \doteq 1$. For $\hat c> \hat c_{SP}$ a second minimum appears. This corresponds to a stable fixed point solution
which form the branch 
$\theta_{\rm mst}\approx \hat c\ln Q +\ln(Q\ln Q)$. Replacing in \eqref{single-col-comp} we find
$\Sigma_{1,1}(\hat c) \approx (1-\frac{\hat c}{2})\ln Q$. This is positive in the interval
$\hat c \in [1,2]$, and looses its meaning beyond $\hat c_s\doteq 2$ which is the static phase transition 
threshold. 
\vskip 0.25cm
\begin{figure}[htp]
\begin{centering}
\input{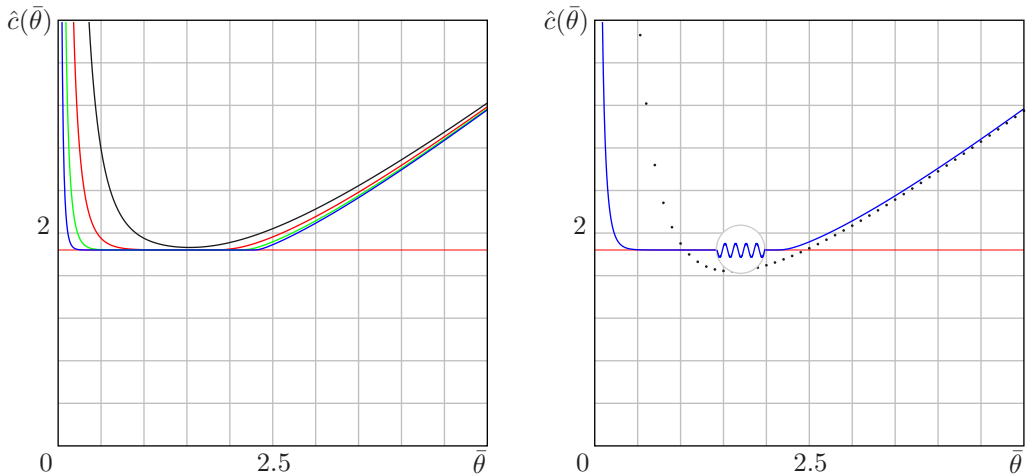}
\caption{{\small {\it Left}: sequence of van der Waals curves for $K=5$, $w=3$ and $L= 10, 20, 40, 80$ (top to bottom). 
They converge to the individual system curve for $\theta\in [\theta_{\rm mst}, +\infty]$. 
{\it Right}: a magnification of the plateau region for $L=40$ shows the fine structure. The dotted curve is the individual system curve 
and the red 
line corresponds to the phase transition threshold $\hat{\alpha}_s=1.840980$.}}
\label{coupled-qcol-curves}
\end{centering}
\end{figure}
\vskip 0.25cm
Let us now turn to the coupled model. The same heuristic arguments than for $K$-SAT hold. In the regime 
$L>>w>>1$ and with asymmetric boundary conditions $\theta_z\to 0$ for $z\to -\infty$ and $\theta\to \theta_{\rm mst}$
for $z\to +\infty$ we take a kink-like ansatz for the solutions of \eqref{one-dim-Q}. Their total average complexity is 
given by a convex combination
of, $\Sigma_{1,1}(\hat c, \theta=0)=0$ and $\Sigma_{1,1}(\hat c, \theta_{\rm mst}) = (1-\frac{\hat c}{2})\ln Q$, with weights 
determined by the location of the kink center $\xi L$. We have 
$\overline\theta \equiv \frac{1}{L}\sum_{z=-\frac{L}{2}+1}^{\frac{L}{2}}\theta_z = (\frac{1}{2} - \xi)\theta_{\rm mst}$
and $\Sigma_{kink}\approx \frac{\overline\theta}{\theta_{\rm mst}} (1-\frac{\hat c}{2})\ln Q$. The stable kink 
fixed point profile corresponds to
$\overline \theta=0$ and $\{\theta_z=0\}$ for all $\hat c<2$ which means that the complexity vanishes for $\hat c <2$. Thus 
within this approximation the SP threshold saturates to the static phase transition threshold.
\vskip 0.25cm
\begin{table}
\centering
\begin{tabular}{c c c c c c}
\hline\hline 
 Q & $5$ & &   $7$  & &   $10$ \\
\hline
$\hat c_{s}$ & $1.840980$ & & $1.911260$ & & $1.635790$ \\
\hline
$\hat{c}_{\text{SP}}$   &  $ 1.6411666$ & & $ 1.651565$ & &  $1.949869$ \\
\hline
$\hat{c}_{\text{SP},80,2}$   &  $1.839709 $ & & $1.906734 $ & &  $1.939527$ \\
$\hat{c}_{\text{SP}, 80, 3}$   &  $1.840978$ & & $1.911213 $ & &  $1.949606$ \\
$\hat{c}_{\text{SP}, 80, 4}$   &  $1.840980$ & & $1.911260 $ & &  $1.949865$ \\
\hline
\end{tabular}
\caption{{\small SP Thresholds of the 
individual $(L=w=1)$ and coupled ensembles ($L= 80, w= 2, 3, 4)$ are found from the van der Waals curves. 
Threshold values $\hat c_s$ are from population dynamics. For $Q=10$ we clearly 
see that the SP threshold saturates to $\hat{c}_s$ from below as $w$ increases.}}
\label{tablelargeq}
\end{table} 
\begin{figure}
\begin{centering}
\input{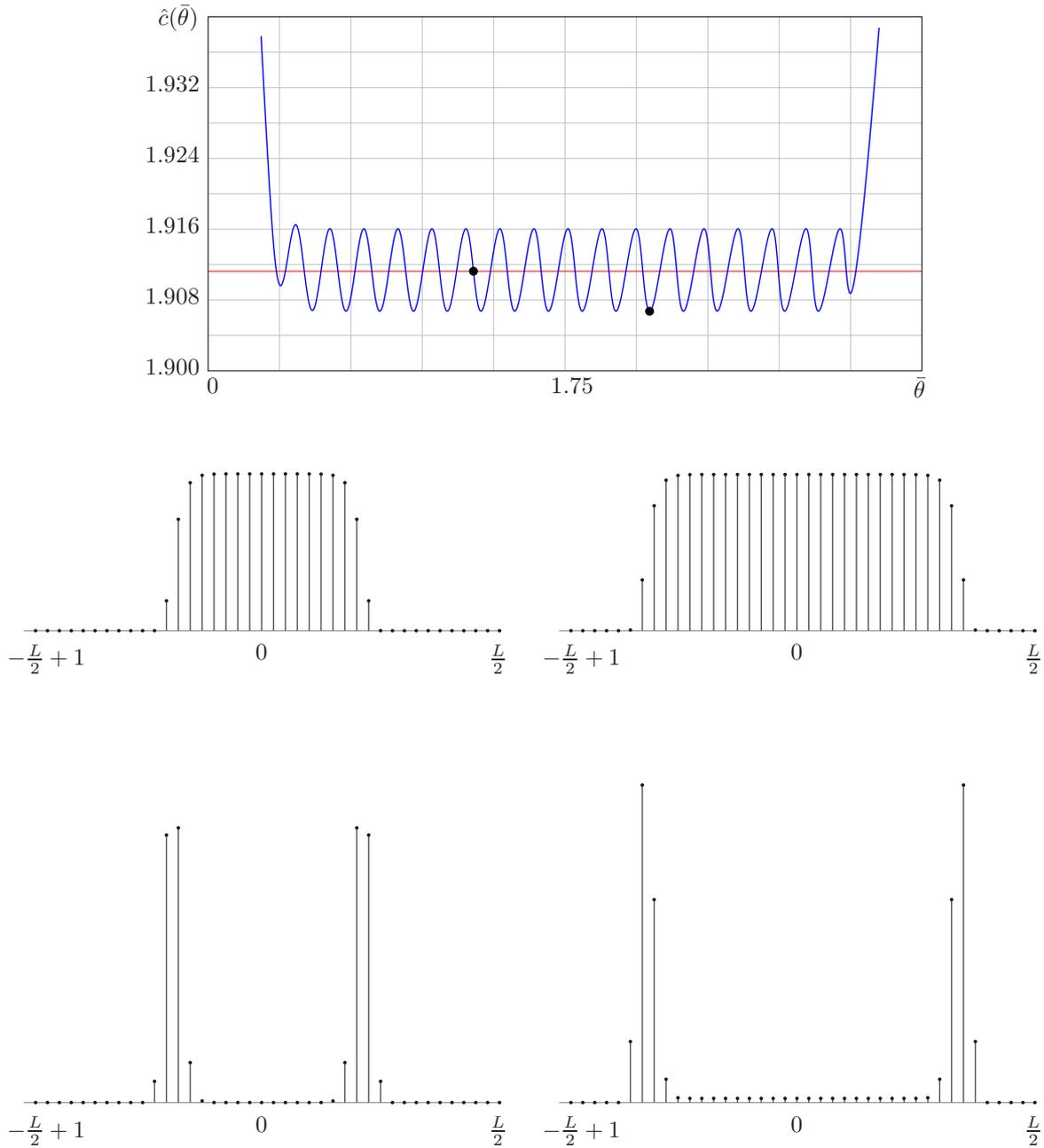}
\caption{{\small van der Waals curves of the  coupled system (top) for $Q=7$, $w=2$ and $L=40$. 
Middle and bottom left are the warning and complexity profiles corresponding to the point with coordinates
 $( \bar{\theta}_l, \hat{c}_l)=(1.30125, 1.91260)$. Notice that $\hat c_l$ is near the phase transition threshold 
so that the complexity nearly vanishes except at the transition regions of the kink. The total average complexity nearly vanishes. 
Middle and bottom right are the warning and complexity density profiles corresponding to 
the point $(\bar{\theta}_r, \hat{c}_r)=(2.165,1.906734)$. For the complexity profile at the bottom right, 
we see that in the middle region of the profile the height of the complexity density is 
$(1-\frac{\hat c_r}{2})\approx 0.047$ and
the total average complexity is 
approximately given by $\frac{\overline\theta_r}{\theta_{\rm mst}} (1-\frac{\hat c_r}{2})\ln Q$.}}
\label{compl-prof-qcol}
\end{centering}
\end{figure}
We solve \eqref{one-dim-Q} numerically with symmetric boundary conditions which enforce the profile to vanish
at the end points of the chain. There exists a family of solution profiles 
$(\hat c(\overline\theta), \{\theta_z(\overline\theta)\})$ parametrized by the total
average warning probability $\overline\theta \equiv \frac{1}{L}\sum_{z=-\frac{L}{2}+1}^{\frac{L}{2}}\theta_z$. Figure
\ref{coupled-qcol-curves} illustrates a sequence of van der Waals curves $\hat c(\overline\theta)$ and Table \ref{tablelargeq}
gives numerical values of their minima which determines $\hat c_{{\rm SP}, w, L}$. 

Finally, Figure \ref{compl-prof-qcol} displays warning and complexity profiles for $\hat c$ in the wiggles region. The results
are analogous to those of $K$-SAT.

\section{Proofs of theorems \ref{therm-limit-1} and \ref{therm-limit-2}}\label{proofs}

In this section we sketch the proofs of theorems \ref{therm-limit-1} and \ref{therm-limit-2}. The proof of theorem \ref{therm-limit-1}
is straightforward and does not depend on the details of the model at hand. On the other hand that of theorem \ref{therm-limit-2}
has to adapted for each model at hand. 

\subsection{Proof of theorem \ref{therm-limit-1}}

Recall that for the Hamiltonian of the open chain $\mathcal{H}_{\rm cou}(\underline x)$ in \eqref{hamiltonian}, $\underline x = (x_{iz})$ with
$(i,z)\in \cup_{z=-\frac{L}{2}+1,\cdots,\frac{L}{2}+w-1} V_z$. It will be convenient to set 
$\underline x=(\underline x^\prime, \underline x^{\prime\prime})$ where $\underline x^\prime = (x_{iz}; z=-\frac{L}{2}+1,\dots, \frac{L}{2})$ and 
$\underline x^{\prime\prime} = (x_{iz}; z=\frac{L}{2}+1,\dots, \frac{L}{2}+w-1)$.
Recall also that the Hamiltonian $\mathcal{H}_{\rm cou}^{\rm per}(\underline x^\prime)$
of the periodic chain is 
given by the same expression  \eqref{hamiltonian} with $\underline x^\prime=(x_{iz}; z=\frac{L}{2}+1,\dots, \frac{L}{2})$. 
Therefore the difference between the two Hamiltonians only comes 
because of the terms $\psi_{cz}(\underline x_{\partial cz})$ with $z=\frac{L}{2}-w+2, \cdots , \frac{L}{2}$.  In other words,
\begin{equation}\label{2Mw}
\vert \mathcal{H}_{\rm cou}(\underline x^\prime, \underline x^{\prime\prime}) - \mathcal{H}_{\rm cou}^{\rm per}(\underline x^\prime)\vert \leq Mw,
\end{equation}
for all $\underline x^{\prime\prime}$.
Now, from the fact that
\begin{equation}
\mathcal{H}_{\rm cou}^{\rm per}(\underline x^\prime) - Mw 
\leq \mathcal{H}_{\rm cou}(\underline x^\prime, \underline x^{\prime\prime}) 
\leq  \mathcal{H}_{\rm cou}^{\rm per}(\underline x^\prime) + Mw,
\end{equation}
and by taking the $\min$, dividing by $NL$ and taking the expectation, we deduce
\begin{equation}
e_{N, L, w}^{\rm per}(\alpha)- \frac{\alpha w}{L} 
\leq e_{N, L, w}(\alpha) 
\leq  e_{N, L, w}^{\rm per}(\alpha) + \frac{\alpha w}{L}, 
\end{equation}
which proves the theorem.

\subsection{Proof of Theorem \ref{therm-limit-2}}\label{interpolations}

As explained in section \ref{generalsetting} the proof that the limit exists, is continuous and non-decreasing, for the individual models
is provided in \cite{BGP09} and is essentially the same for the coupled periodic chain. Here we prove the equality
of the two limits \eqref{equality-two-limits}. 
The following notation is convenient. For a given graph instance $G$ (from some ensemble) 
we call $\mathcal{H}_G(\underline x)$ the 
corresponding Hamiltonian. It always consists, as in \eqref{hamiltonian}, of a sum of 
terms $1-\psi_{cz}(x_{\partial cz})$ over constraints $(c,z)\in G$ . The ground state energy is equal to $\min_{\underline x}\mathcal{H}_G(\underline x)$.
To set up suitable interpolation procedures, it is convenient to first
define three extra ensembles.

\vskip 0.15cm
\noindent{\it The ``connected'' ensemble}. This is essentially the individual $[N,K,\alpha]$ ensemble scaled by $L$. 
We have a set of $LN$ variable nodes
and a set of $LM$ constraint nodes. Each constraint node has $K$ edges connected u.a.r. to variable nodes. 
Expectations with respect to this ensemble are denoted by
$\mathbb{E}_{\rm conn}$. Because of the existence of the limit we have  
\begin{equation}
\lim_{N\to +\infty}\frac{1}{LN}\mathbb{E}_{\rm conn}[\min_{\underline x}\mathcal{H}_G(\underline x)] 
= \lim_{N\to +\infty}e_N(\alpha),
\label{first}
\end{equation}
for any fixed $L$.

\vskip 0.15cm
\noindent{\it The ``disconnected'' ensemble}. This is a variant of 
the individual $[N,K,\alpha]$ ensemble replicated $L$ times. We place at positions $z= -\frac{L}{2}+1,\ldots,\frac{L}{2}$, $L$ disjoint sets of variable nodes 
$V_z$ containing each $N$ nodes. Each node from the set of $LM$ constraint nodes is affected u.a.r. to a position 
$z= -\frac{L}{2}+1,\ldots,\frac{L}{2}$. Note that the set $\tilde C_z$ of constraint nodes at position $z$ has cardinality $M_z\sim {\rm Bi}(LM, \frac{1}{L})$. 
Each node from $\tilde C_z$ has $K$ edges that are
connected u.a.r. to nodes in $V_z$.
Expectations are denoted by $\mathbb{E}_{\rm disc}$.
Since each $M_z$ is concentrated on $M$ with a fluctuation $O(\sqrt M)$, we can show by an argument similar to the 
proof of theorem \ref{therm-limit-1} that 
\begin{equation}
\frac{1}{LN}\mathbb{E}_{\rm disc}[\min_{\underline x}\mathcal{H}_G(\underline x)] = e_N(\alpha) + O(N^{-1/2}),
\label{second}
\end{equation}
where $O(N^{-1/2})$ is uniform in $L$.

\vskip 0.15cm
\noindent{\it The ``ring'' ensemble}. This is a variant of the periodic chain in section \ref{generalsetting}.
We place at positions $z= -\frac{L}{2}+1,\ldots,\frac{L}{2}$, $L$ disjoint sets of variable nodes 
$V_z$, each containing $N$ nodes. 
Now we have a set of $LM$ constraint nodes. Each constraint node is affected to a position $z= -\frac{L}{2}+1,\ldots,\frac{L}{2}$ u.a.r.
and (say the position is $z$) its $K$ edges are connected u.a.r. to the set of variables $\cup_{k=0}^{w-1}{V}_{z+k \mod L}$. Note that the sets
$\tilde C_z$ of constraint nodes have cardinalities $M_z\sim {\rm Bi}(LM, \frac{1}{L})$.
We denote by $\mathbb{E}_{\rm ring}$ the expectation with respect to this ensemble.
Since each $M_z$ is concentrated on $M$ with a fluctuation $O(\sqrt M)$, an argument similar to the 
proof of theorem \ref{therm-limit-1} shows that 
\begin{equation}
 \frac{1}{LN}\mathbb{E}_{\rm ring}[\min_{\underline x}\mathcal{H}_G(\underline x)] = e_{N, L, w}^{\rm per}(\alpha) + O(N^{-1/2}),
 \label{third}
\end{equation}
where $O(N^{-1/2})$ is uniform in $L$ (and depends on $w$).
\vskip 0.15cm
We will show
\begin{equation}
\mathbb{E}_{\rm conn}[\min_{\underline x}\mathcal{H}_G(\underline x)] 
\leq 
\mathbb{E}_{\rm ring}[\min_{\underline x}\mathcal{H}_G(\underline x)] 
\leq 
\mathbb{E}_{\rm disc}[\min_{\underline x}\mathcal{H}_G(\underline x)],
\label{two-bounds}
\end{equation}
which allows to conclude the proof of the theorem by using \eqref{first}, \eqref{second}, \eqref{third}.
\vskip 0.15cm
\noindent{\it Left inequality in \eqref{two-bounds}}. We 
build a sequence of interpolating ``$r$-ensembles'', $r=0,\ldots, LM$, interpolating between the ring ($r=0$) and connected 
($r=LM$) ensembles.  We have two sets of $LM$ constraint and $LN$ variable nodes. The variable nodes 
are organized into $L$ disjoint sets $V_z$ each containing $N$ nodes, placed along the positions $z= -\frac{L}{2}+1,\ldots,\frac{L}{2}$.
Expectation with respect to the $r$-ensemble is denoted by $\mathbb{E}_r$. To sample a graph $G_r$ from this 
ensemble we first take $r$ nodes - called type 1 - from the set of $LM$ constraint nodes. Each one has $K$ edges which are connected
u.a.r. to the set of $LN$ variable nodes. For the remaining $LM-r$ constraint nodes - called type 2 - we proceed as follows:
each one is affected u.a.r. to a position $z$, and its $K$ edges are then connected u.a.r. to the $wN$ variable nodes 
in $\cup_{k=0}^{w-1}{V}_{z+k \mod L}$.
We claim that for $1 \leq r \leq LM $,
\begin{align}
\mathbb{E}_r[\min_{\underline x}\mathcal{H}_{G_r}(\underline x)] \leq 
\mathbb{E}_{r-1}[\min_{\underline x}\mathcal{H}_{G_{r-1}}(\underline x)].
\label{claim1}
\end{align}
Clearly this implies the left inequality in \eqref{two-bounds}.
Let us prove this claim. Take a random graph $G_r$ and delete u.a.r. a constraint from the type 1 nodes: this yields an intermediate graph
$\tilde G$. One can go back to a random graph $G_r$ by adding back a type 1 node according to the above rules, or one can go to
a random graph $G_{r-1}$ by adding back a type 2 node according to the above rules. 
We will prove that conditioned on any realization of $\tilde G$ we have
\begin{align}
\mathbb{E}_r[\min_{\underline x}\mathcal{H}_{G_r}(\underline x) \mid \tilde G] \leq 
\mathbb{E}_{r-1}[\min_{\underline x}\mathcal{H}_{G_{r-1}}(\underline x) \mid \tilde G].
\label{claim4}
\end{align}
Claim \eqref{claim1} follows by averaging over 
$\tilde G$. We now prove \eqref{claim4}  for K-SAT and Q-coloring separately.

\noindent{\underline{K-SAT}}: Consider the set of "optimal assignments" $\underline x$ that minimize
$\mathcal{H}_{\tilde G}(\underline x)$. We say that a variable is {\it frozen} iff
it takes the same value for all optimal assignments. We call $\mathcal{F}$ the set of variable nodes with frozen variables and $\mathcal{F}_z = \mathcal{F}\cap V_z$.
Now consider adding a new constraint node $n$ to the graph $\tilde G$. This will cost an extra energy iff the 
node $n$ connects only to frozen variable nodes {\it and} does not satisfy them. For such an event we have 
\begin{equation}
\min_{\underline x}\mathcal{H}_{\tilde G\cup n}(\underline x)-\min_{\underline x}\mathcal{H}_{\tilde G}(\underline x)=1.
\end{equation}
When the node $n$ is connected u.a.r. to the $LN$ variable nodes ($n$ is type 1) this event has probability
$\frac{1}{2^K}(\frac{\vert \mathcal{F} \vert}{LN})^K$.
Thus
\begin{align}
\mathbb{E}_r[\min_{\underline x}\mathcal{H}_{G_r}(\underline x) \mid \tilde G] - 
\min_{\underline x}\mathcal{H}_{\tilde G}(\underline x)
= \frac{1}{2^K} \biggl(\frac{\vert \mathcal{F} \vert}{LN}\biggr)^K.
\end{align}
Similarly when the node $n$ is affected u.a.r. to a position $z$ and then connected u.a.r. to 
$\cup_{k=0}^{w-1}{V}_{z+k \mod L}$ ($n$ is type 2) we get
\begin{align}
\mathbb{E}_{r-1}[\min_{\underline x}\mathcal{H}_{G_{r-1}}(\underline x) \mid \tilde G] - 
\min_{\underline x}\mathcal{H}_{\tilde G}(\underline x)
= 
\frac{1}{L}\sum_{z=-\frac{L}{2}+1}^{\frac{L}{2}}\frac{1}{2^K} 
\bigl(\frac{1}{wN}\sum_{k=0 }^{w-1} \vert \mathcal{F}_{z+k{\rm mod} L} \vert\bigr)^K.
\end{align}
Claim \eqref{claim4} follows from the last two equations, convexity of $x^K$ for $x\geq 0$, and 
\begin{equation}
\vert\mathcal{F}\vert= 
\sum_{z=-\frac{L}{2}+1}^{\frac{L}{2}}\frac{1}{w}\sum_{k=0 }^{w-1} \vert \mathcal{F}_{z+k\mod L} \vert.
\end{equation} 

\noindent{\underline{Q-Coloring}}: The proof is similar. Consider the set of "optimal colorings" that minimize $\mathcal{H}_{\tilde G}(\underline x)$. 
We define an equivalence relation between variable nodes: we say that two nodes are equivalent iff  
their two colors are identical for all optimal assignments. Let
$\mathcal{F}^j$, $1 \leq j \leq J$, be the equivalence classes of 
nodes and let $\mathcal{F}_z^j=\mathcal{F}^j \cap \mathcal{V}_z$. Now, assume we add a random constraint node
$n$ 
to $\tilde G$. We have 
$\min_{\underline x}\mathcal{H}_{\tilde G\cup n}(\underline x)
-\min_{\underline x}\mathcal{H}_{\tilde G}(\underline x)=1$ only when  $n$ chooses its two variables from the same equivalence class; otherwise 
the energy difference is zero. Thus we obtain
\begin{align}
\mathbb{E}_r[\min_{\underline x}\mathcal{H}_{G_r}(\underline x) \mid \tilde G ] - 
\min_{\underline x}\mathcal{H}_{\tilde G}(\underline x)
=  \sum_{i=1}^J\biggl(\frac{\vert \mathcal{F}^j \vert}{LN}\biggr)^2,
\end{align}
and 
\begin{align}
\mathbb{E}_{r-1}[\min_{\underline x}\mathcal{H}_{G_{r-1}}(\underline x) \mid \tilde G ] - 
\min_{\underline x}\mathcal{H}_{\tilde G}(\underline x)= 
\frac{1}{L}\sum_{z=-\frac{L}{2}}^{\frac{L}{2}}\sum_{j=1}^J \bigl(\frac{1}{wN} \sum_{k=0 }^{w-1} 
\vert \mathcal{F}_{z + k\mod L}^j \vert \bigr)^2.
\end{align}  
Claim \eqref{claim4} then follows from the last two equations and convexity of $x^2$. 

\vskip 0.15cm 

\noindent{\it Right inequality in \eqref{two-bounds}}.
We construct new $r$-ensembles, $r=0,\ldots,LM$ that now interpolate between the disconnected ($r=0$) and the ring ($r=LM$) ensembles.
A random graph $G_r$ is constructed as follows. We have a set of $LM$ constraint nodes and a set of $LN$ variable nodes 
organized into $L$ disjoint sets $V_z$ each containing $N$ nodes, placed along positions $z$. We first take $r$ constraint nodes, called type 1.
Each of them is affected u.a.r. to a position $z$, and its $K$ edges are connected u.a.r. to variable nodes in $V_z$.
Next, the remaining $LM-r$ constraints 
nodes - called type 2 - are each affected u.a.r. to a position $z$ and its $K$ edges are 
connected u.a.r. to $wN$ nodes in $\cup_{k=0}^{w-1} V_{z+k\mod L}$.  Note that at each position there are ${\rm Bi}(r, \frac{1}{L})$
type 1 nodes and ${\rm Bi}(LM-r,\frac{1}{L})$ type 2 nodes, so in total there are ${\rm Bi}(LM, \frac{1}{L})$ constraint nodes. 
Similarly to the previous interpolation we will prove 
\begin{align}
\mathbb{E}_{r}[\min_{\underline x}\mathcal{H}_{G_r}(\underline x)] 
\leq 
\mathbb{E}_{r-1}[\min_{\underline x}\mathcal{H}_{G_{r-1}}(\underline x)]. 
\label{as-before}
\end{align} 
This inequality implies the upper bound in \eqref{two-bounds}.
To prove \eqref{as-before}, as before, we consider the random graph $\tilde G$ obtained by deleting u.a.r. a type 1 node from $G_r$.
From $\tilde G$ one gets a random graph $G_r$ by adding back a type 1 node, or one gets a graph $G_{r-1}$ by adding back a type 2 node instead.
We first prove that 
\begin{align}
\mathbb{E}_r[\min_{\underline x}\mathcal{H}_{G_r}(\underline x) \mid \tilde G] \leq 
\mathbb{E}_{r-1}[\min_{\underline x}\mathcal{H}_{G_{r-1}}(\underline x) \mid \tilde G],
\label{claim5}
\end{align}
and then by averaging over graphs $\tilde G$ we get \eqref{as-before}.
Let us briefly sketch the derivation of \eqref{claim5}.

\noindent{\underline{$K$-SAT}}: 
We use the same sets $\mathcal{F}_z$ of frozen variables at position $z$ corresponding to the ground state configurations of $\mathcal{H}_{\tilde G}(\underline x)$. 
We have 
\begin{equation}
\mathbb{E}_{r}[\min_{\underline x}\mathcal{H}_{G_{r}}(\underline x) \mid \tilde G]  - \min_{\underline x}\mathcal{H}_{\tilde G}(\underline x) 
= 
\frac{1}{L} \sum_{z=-\frac{L}{2}+1}^{\frac{L}{2}} \frac{1}{2^K} \biggl(\frac{\vert \mathcal{F}_z \vert}{N}\biggr)^K,
\end{equation}
and
\begin{align}
\mathbb{E}_{r-1}[\min_{\underline x}\mathcal{H}_{G_{r-1}}(\underline x) \mid \tilde G]  - \min_{\underline x}\mathcal{H}_{\tilde G}(\underline x) 
= \frac{1}{L}\sum_{z=-\frac{L}{2}+1}^{\frac{L}{2}}\frac{1}{2^K} \bigl(\frac{1}{wN} \sum_{k=0 }^{w-1} 
\vert \mathcal{F}_{z+k{\rm mod} L} \vert \bigr)^K.
\end{align}
Estimate \eqref{claim5} then follows by the convexity of the function $x^K$ for $x\geq 0$.

\noindent{\underline{$Q$-coloring}}: We use the same equivalence relation between variable nodes and sets $\mathcal{F}^j_z$. We have 
\begin{align}
\mathbb{E}_r[\min_{\underline x}\mathcal{H}_{G_r}(\underline x) \mid \tilde G ] - 
\min_{\underline x}\mathcal{H}_{\tilde G}(\underline x)
=  \frac{1}{L}\sum_{z=-\frac{L}{2}+1}^{\frac{L}{2}}\sum_{i=1}^J\biggl(\frac{\vert \mathcal{F}^j_z \vert}{N}\biggr)^2,
\end{align}
and 
\begin{align}
\mathbb{E}_{r-1}[\min_{\underline x}\mathcal{H}_{G_{r-1}}(\underline x) \mid \tilde G ] - 
\min_{\underline x}\mathcal{H}_{\tilde G}(\underline x)= 
\frac{1}{L}\sum_{z=-\frac{L}{2}+1}^{\frac{L}{2}}\sum_{j=1}^J \bigl(\frac{1}{wN} \sum_{k=0 }^{w-1} 
\vert \mathcal{F}_{z + k {\rm mod} L}^j \vert \bigr)^2.
\end{align}  
Again, estimate \eqref{claim5} then follows from the convexity of $x^2$. 

\section{Dynamical and condensation thresholds}\label{dyncond}

The SP formalism says nothing about the relative sizes (internal entropy) of clusters of solutions and consequently does not take into account 
which of them are "relevant" to the uniform measure over zero energy solutions. For similar reasons, it is not clear that the SP threshold has particular algorithmic significance.
These issues are partly addressed by the more elaborate 
entropic cavity method \cite{entropiccavity}, \cite{lenka}, \cite{elaborate}. 
It predicts the existence of the dynamical and condensation thresholds $\alpha_d$ and $\alpha_c$.
The dynamical threshold is believed to separate a phase ($\alpha<\alpha_d$) where the uniform measure is essentially supported on one well 
connected cluster of dominant entropy, and a phase ($\alpha_d<\alpha<\alpha_c$) where the measure is supported on an exponential number of clusters 
with equal internal entropy. For $\alpha>\alpha_c$ the measure condenses on a "handful" of clusters of dominant entropy. The condensation 
threshold is a static thermodynamic transition in the sense that the total ground state entropy has a non-analyticity
as a function of $\alpha$. The algorithmic significance of $\alpha_d$ and/or $\alpha_c$ is still unclear. But 
see \cite{ricci-tersenghi-semerjian}, \cite{coja} 
for recent related results.  

We have computed the dynamical and condensation thresholds of coupled CSP. Let us denote them 
$\alpha_{d,L,w}$ and $\alpha_{c,L,w}$ (with $w$ fixed). We observe that as $L$ increases 
$\alpha_{c,L,w}\to \alpha_c$. This observation is consistent with the following rigorous result
that we prove in appendix \ref{A}:  the thermodynamic limit of the free energy (at finite temperature) of the chain is identical 
to that of the individual model. From the free energy one can formally obtain the entropy by differentiating the free energy with respect to temperature. 
The result about the free energy then suggests that the zero temperature entropy of the chain and 
individual models have the same non-analyticity points as a function of the constraint density. The second 
important observation is that in the regime $1<<w<<L$ we find $\alpha_{d, L,w}\to \alpha_c$. Thus the dynamical 
threshold saturates towards the condensation threshold\footnote{Note that for $K=3$ we already have $\alpha_d=\alpha_c$ for the individual ensemble.}.

The dynamical and condensation thresholds are analogous to the dynamical and condensation temperatures of $p$-spin glass models 
for $p\geq 3$, and to the glassy and Kauzmann transition temperatures in structural glasses \cite{Wolynes}, \cite{Kirkpatrick}, \cite{biroli}. 
One expects that a similar saturation of the dynamical towards the condensation 
temperature holds for coupled $p$-spin glass models on complete graphs for $p\geq 3$. On the other hand, for $p=2$ the 
replica symmetry breaking transition is continuous, there is no dynamical temperature, and spatial coupling is not expected to modify the phase diagram.

Table \ref{table-general-picture} summarizes all the behaviors of the SP, SAT-UNSAT, dynamical and condensation thresholds
for the $K$-SAT problem. The situation for coloring is similar.

\begin{table}
\centering
\begin{tabular}{c c c c c c c c c c c c c c c c}
\hline\hline 
 $K$ && $\alpha_{SP}$&$\alpha_{SP,80,3}$ &  & $\alpha_d$&$\alpha_{d,80,3}$  &  & $\alpha_c$&$\alpha_{c,80,3}$ &  & $\alpha_s$&$\alpha_{s,80,3}$ \\
\hline
$3$ && $3.927$&$4.268$ &  & $3.86$&$3.86$ &  & $3.86$&$3.86$ &  & $4.267$&$4.268$  \\ 
\hline
$4$ && $8.30$ &$9.94$ &  & $9.38$ &$9.55 $ &  & $9.55$ &$9.56$ &  & $9.93 $& $10.06 $  \\
\hline
\end{tabular}
\caption{{\small Thresholds of individual and coupled $K$-SAT model for 
$L=80$ and $w=3$. The condensation and SAT-UNSAT thresholds correspond to non analyticities of the entropy and ground 
state energy and remain unchanged (for $L\to +\infty$). Already for $w=3$ the dynamical and SP thresholds saturate very close 
to $\alpha_c$ and $\alpha_s$.}}
\label{table-general-picture}
\end{table}

\section{An application of threshold saturation to algorithmic lower bounds}\label{last-section}

We briefly discuss a methodology, that uses coupled CSP ensembles, for proving lower bounds on the static phase transition threshold of individual CSP ensembles.
We illustrate it with simple examples and show how threshold saturation can help. These examples do not reach the best known
lower bounds,  but they serve well to illustrate a new methodology for attacking the problem. We keep the discussion at an informal level.

Given a CSP from an individual ensemble, one usually tries to devise an algorithm that provably
finds solutions w.h.p for $\alpha<\alpha_{\rm alg}$. This then implies 
$\alpha_{\rm alg} < \alpha_s$. Consider now the coupled ensemble, and apply the same algorithm.
Call $\alpha_{\rm alg, L, w}$ the algorithmic threshold for 
the existence w.h.p of solutions and set $\alpha_{{\rm alg}, w} = \lim_{L\to +\infty}\alpha_{\rm alg, L, w}$.
From theorems \ref{therm-limit-1} and \ref{therm-limit-2} we know that the coupled ensemble has the same static phase transition threshold
as the individual one, when $L\to +\infty$. Therefore one certainly has  the lower bound $\alpha_{\rm alg, w}< \alpha_s$.
The point is that for well chosen algorithms an improvement of the bound may occur, namely  
$\alpha_{\rm alg}<\alpha_{\rm alg, w}< \alpha_s$, and one would expect to get the best lower bounds by increasing $w$. 
A well chosen algorithm is one that shows a ''threshold improvement`` or even saturation phenomenon.
Somehow the ''seed`` provided by the reduced hardness near the boundaries should grow and propagate in the bulk.

Below we illustrate the idea with three simple peeling algorithms applied to $K$-XORSAT, $K$-SAT and $Q$-COL. 

\underline{$K$-XORSAT}. This case provides the best illustration.  
The individual model has a static phase transition at $\alpha_s$, and a clustering transition at $\alpha_{SP}$ with a complexity
counting clusters in the interval $[\alpha_{SP}, \alpha_s]$.
In appendix \ref{A} (theorem \ref{free-xorsat-limit}) we show that 
  the coupled and 
individual ensembles have the same phase transition threshold $\alpha_s$ for even $K$ (for odd $K$ the proof breaks down but the result is presumably true). 
Now consider the ''leaf removal`` algorithm. As long as there is a leaf variable node remove it, and remove the attached constraint node with its emanating edges.
If this process ends with an empty graph the instance is satisfiable. It is known that this algorithm is equivalent
to BP message passing, and the density evolution analysis 
leads to the fixed point equation
\begin{equation}
 x=(1-\exp(-\alpha K x)^{K-1}.
\label{xorsat1}
\end{equation}
Here, $x$ is interpreted as the probability (when the number of iterations goes to infinity) that a constraint node is not removed. 
There is a threshold $\alpha_{\rm lr}$
above which \eqref{xorsat1} has non-trivial fixed points (i.e, the fraction of remaining variables is positive), so we get a lower bound $\alpha_{\rm lr}<\alpha_s$. For the coupled ensemble
the density evolution analysis yields the one-dimensional fixed point equations
\begin{equation}
 x_z = \biggl\{\frac{1}{w}\sum_{l=0}^{w-1}
(1- \exp(-\frac{\alpha K}{w} \sum_{k=0}^{w-1} x_{z+k-l})
\biggr\}^{K-1},
\label{xorsat2}
\end{equation}
with boundary condition $x_z=0$ for $z$ at the boundaries. Solving for the non-trivial kink solutions numerically, we indeed observe 
$\alpha_{\rm lr}<\alpha_{\rm lr, w} <\alpha_s$. Table \ref{table-xor} shows the threshold improvement for $w=5$ and the first few values of $K$. In fact 
the leaf removal threshold even saturates $\alpha_{{\rm lr}, w}\uparrow \alpha_s$ as $w\to +\infty$. This is not surprising since for XORSAT the SP formalism leads to
the same fixed point equations (but with a different interpretation for $x$ and $x_z$) \cite{Mezard-Montanari-book}, \cite{xor-sat}. In particular there is 
a complexity $\Sigma_{\rm xorsat}(\alpha)>0$ for $\alpha\in [\alpha_{SP}, \alpha_s]$ counting the number of clusters of solutions in Hamming space with
$\alpha_{SP}=\alpha_{\rm lr}$ and $\Sigma_{\rm xorsat}(\alpha_s)=0$. Note that for large $K$ one finds $\alpha_{\rm lr} = \ln K/K + \ln\ln K/K + 1/K + o(1/K)$
and $\alpha_s= 1+o(1)$.
\begin{table}
\centering
\begin{tabular}{c c c c c c c c c  }
\hline\hline 
 K & $3$ & & $4$  & & $5$ & & $7$   \\
\hline
$\alpha_s$           &  $0.917$    & &  $0.976$    & &   $0.992$       &&  $0.999$     \\
\hline
$\alpha_{\text{lr}}$ & $0.818$ & & $0.772$ & & $0.701$ & & $0.595$  \\
\hline
$\alpha_{\text{lr, L=80, w=5}}$   &  $0.917$ & & $0.977$ & &  $0.992$ & &$0.999$  \\
\hline
\end{tabular}
\caption{{\small {\it First line}: phase transition threshold for $K$-XORSAT. 
{\it Second line}: leaf removal threshold for the the uncoupled case. {\it Third line}:  leaf removal threshold for a coupled chain with $w=5$, $L=80$.}}
\label{table-xor}
\end{table} 

\underline{$K$-SAT}.
Let us now turn to $K$-SAT and consider the ''pure literal rule`` algorithm \cite{Franco}, \cite{broder-frieze}. Consider variable nodes that have only one type 
of edge - dashed or full - attached to them. As long as there are such nodes (called ''pure``) set the variable to the value which satisfies all the attached 
constraints and remove these constraints and their edges. Continue until no ''pure`` node remains. If no constraint node remains then 
the algorithm succeeds in finding a satisfying assignment. 
This algorithm can be 
cast in a message passing form and can be analyzed by the density evolution method \cite{luby-etal}. The net result is that the pure literal rule
 succeeds w.h.p for $\alpha<\alpha_{\rm pl}$ such that 
\begin{equation}
 x=(1- \exp(-\frac{\alpha K}{2} x))^{K-1}\,.
 \label{pure-literal-fixed-point}
\end{equation}
has a unique fixed point $x=0$. 
We now take coupled instances from the ensemble defined in section \ref{generalsetting}. In order to analyze the pure literal rule 
we can think of extending the chain to $\mathbb{Z}$ with "pure" variable nodes for $z\leq -\frac{L}{2}$ and $\geq \frac{L}{2} + w$. 
The peeling of constraints attached to pure nodes will propagate inside the chain as long as $\alpha$ is not too large.
The analysis yields the one-dimensional fixed point equations
\begin{equation}
 x_z = \biggl\{\frac{1}{w}\sum_{l=0}^{w-1}
(1- \exp(-\frac{\alpha K}{2w} \sum_{k=0}^{w-1} x_{z+k-l})
\biggr\}^{K-1}
\label{pure-literal-fixed-point-cou}
\end{equation}
with boundary condition $x_z=0$ for $z$ at the boundaries.
Note that \eqref{pure-literal-fixed-point}, \eqref{pure-literal-fixed-point-cou} are the same as \eqref{xorsat1}, \eqref{xorsat2} 
with the replacement $\alpha\to \alpha/2$.
Therefore the pure literal thresholds $\alpha_{\rm pl}$, $\alpha_{{\rm pl}, w}$ for the individual and coupled ensembles are obtained just by doubling
the XORSAT thresholds. For example for $K=3$ we have $\alpha_{\rm pl}\approx 1.636 <\alpha_{\rm pl, w=5, L=80}\approx1.835< \alpha_s\approx 4.26$, a modest improvement.
Interestingly when $K\to +\infty$ we have 
\begin{equation}
\alpha_{\rm pl}\doteq \frac{2\ln K}{K} {\rm ~~but~~} \alpha_{\rm pl, w}\to 2 {\rm ~~ as ~~} w\to +\infty
\end{equation}
Of course this is still a ridiculous lower
bound since we know that $\alpha_s\doteq 2^K\ln 2$. 

\underline{$Q$-COL}. Finally we discuss a similar peeling algorithm for $Q$-COL. This algorithm determines the $Q$-core of a graph $G$ and has been analyzed by the method of differential equations
\cite{Pittel}. Here we discuss the algorithm from the message passing point of view.
Assume there exists a node $i$ in $G$ that has degree less than $Q$. Clearly, if we can color the graph $G\setminus i$ with $Q$ colors, then $G$ can 
also be colored with $Q$ colors. Hence, finding a $Q$-coloring for $G$ is equivalent to finding a $Q$-coloring for $G\setminus i$. As a result,  
we can peel the node $i$ from $G$ and continue this process until the final graph (the $Q$-core) has no more nodes of degree less than $Q$. 
If the final graph is empty then the algorithm succeeds otherwise it fails.
 
Now, consider the following message passing rule. At time $t \in \{1,2, \cdots \}$, assign to each edge 
$\langle i,j\rangle\in E$ two messages $\mu_{i \to j}^t$ and $\mu_{j \to i}^t$. The messages at time 
$t+1$ are evolved from the ones at time $t$ via the following procedure:  
\begin{enumerate}
 \item At time $0$, initialize all the messages to $0$.
\item At time $t+1$,
\begin{equation*}
 \mu_{i\to j}^{t+1}= \mathbbm{1}(\sum_{h \in \partial i \setminus j} \mu_{h\to i}^t < Q-1).
\end{equation*}
\end{enumerate}
The above message passing rule is equivalent to the peeling algorithm in the sense 
that when $\mu_{i \to j}^t=1$, the vertex $i$ would have been peeled by the algorithm some time before $t$ and if $\mu_{i \to j}^t=0$, 
the vertex $i$ would not have been peeled by the algorithm up to time $t$.

Define $x_t=\mathbb{P}(\mu_{i \to j}^t=1)$. we derive the density evolution equation that relates $x_{t+1}$ to $x_t$. 
Let $G$ be randomly chosen from $G(N,\frac{c}{N})$ 
with $N$ very large.
Fix an edge $\langle i, j\rangle$. Observe that $\mu_{i\to j}^{t+1}=1$  if and only if the number of incoming
 messages  that have value $1$ is $\leq Q-2$. Moreover, the probability
that the degree of $i$ is equal to $d\geq 1$ is $e^{-c} \frac{c^{d-1}}{(d-1)!}$. Hence, we can write,
\begin{equation}
 x_{t+1} = \sum_{d=1}^{Q-1} e^{-c} \frac{c^{d-1}}{(d-1)!} + \sum_{d=Q}^\infty e^{-c} \frac{c^{d-1}}{(d-1)!} \sum_{j=0}^{Q-2} {d-1 \choose j} (1-x_t)^jx_t^{d-1-j}.
\end{equation}
One can simplify this equation. Indeed,
\begin{align} \label{peeling-col}
 1- x_{t+1} & = \sum_{d=Q}^\infty e^{-c} \frac{c^{d-1}}{(d-1)!} \biggl\{1- \sum_{j=0}^{Q-2} {d-1 \choose j} (1-x_t)^jx_t^{d-1-j}\biggr\}
 \nonumber \\ &
  = \sum_{d=Q}^\infty e^{-c} \frac{c^{d-1}}{(d-1)!} \sum_{j=Q-1}^{d-1} {d-1 \choose j} (1-x_t)^jx_t^{d-1-j}
 \nonumber \\ &
 = e^{-c} \sum_{j=Q-1}^\infty \sum_{d=j+1}^{+\infty} \frac{(c(1-x_t))^j}{j!} \frac{(cx_t)^{d-1-j}}{(d-1-j)!}
 \nonumber \\ &
 = e^{-c(1-x_t)}\sum_{j=Q-1}^\infty \frac{(c(1-x_t))^j}{j!}
 \nonumber \\ &
 = 1 - e^{-c(1-x_t)}\sum_{j=0}^{Q-2} \frac{c^j}{j!}(1-x_t)^j.
\end{align}
Defining $y \equiv c(1-x)$ and  the function  
\begin{equation}
G(y) = 1-e^{-y}\sum_{j=0}^{Q-2} \frac{y^j}{j!},
\end{equation}
we see that we have to study the solutions of the fixed point equation 
\begin{equation}
y=c \,G(y).
\end{equation}
For $c<c_p$ there is a unique trivial fixed point $y=0$ (i.e., $x=1$) and the algorithm succeeds. Non trivial fixed points appear 
for $c>c_p$ which is the threshold for the emergence of a $Q$-core. 
Table~\ref{table-peeling} contains the numerical values of $c_{\text{p}}$ for several values of $Q$. 

We now take coupled instances from the ensemble defined in section \ref{generalsetting}. 
We can write down the density evolution equations and the corresponding one-dimensional fixed point equations. Not surprisingly, 
similar calculations show that the message passing algorithm is controlled by the one-dimensional fixed point equation,
\begin{equation}
y_z = c \,G\bigl(\frac{1}{2w-1}\sum_{k=-w+1}^{w-1} y_{z+k}\bigr).
\end{equation}
where $y_z=c(1-x_z)$ and $x_z$ is the fraction of peeled nodes at position $z$.
Table~\ref{table-peeling} contains the numerical values of $c_{\text{p,w=5,L=80}}$ for several values of $Q$, and 
shows the threshold improvement. It can be checked numerically that $c_p\doteq Q$ and $c_{p,w,L}\doteq 2Q$ 
for $1<<w<<L$. This has to be compared with $c_s\doteq 2Q\ln Q$. 

\begin{table}
\centering
\begin{tabular}{c c c c c c c c c  }
\hline\hline 
 Q & $3$ & & $4$  & & $5$ & & $7$   \\
\hline
$c_{s}$ & $4.69$ & & $8.90$ & & $13.69$ & & $24.46$  \\ 
\hline
$c_{\text{p}}$ & $3.35$ & & $5.14$ & & $6.79$ & & $9.87$  \\
\hline
$c_{\text{p, L=80, w=5}}$   &  $3.58$ & & $5.74$ & &  $7.84$ & &$11.92$  \\
\hline
\end{tabular}
\caption{{\small {\it First line}: static phase transition threshold for $Q$-COL. 
{\it Second line}: peeling algorithm threshold for the the uncoupled case. {\it Third line}:  peeling algorithm threshold for a coupled chain with $w=5$, $L=80$.}}
\label{table-peeling}
\end{table} 

\section{Conclusion}

In this work we have developed in detail the SP formalism for coupled CSP. We find that the SP thresholds of spatially coupled 
random K-satisfiability and Q-coloring ensembles nicely saturate towards the SAT-UNSAT phase transition threshold of the individual ensemble. 
Moreover the SAT-UNSAT phase transition threshold of the coupled and individual ensembles are identical as required 
by theorems \ref{therm-limit-1} and \ref{therm-limit-2}. 
The  saturation of the SP threshold is remarkably similar to the one of the Belief Propagation algorithmic threshold (towards the optimal one 
associated to the Maximum a Posteriori decoder)
observed in coding theory.  

Let us point out a few issues that would deserve more investigations, and to which we hope to come back in the future.

The large $K$ and $Q$ analysis has shown that when $\alpha$ is in a small interval where the zero-energy complexity 
is strictly positive, the warning and complexity densities form kink-like profiles. These are very similar to the 
kink-like magnetization and free energy densities found in the CW chain. A possible interpretation of the 
complexity density profiles is that the clusters do not only have a "size" given by their internal entropy but also have a "shape" that could 
be taken into account by an extension of the entropic cavity method. The simplest system where this issue could be elucidated is the XOR-SAT system for which the clusters can be precisely defined \cite{xor-sat}. We hope to come back to this question in the near future.

As briefly discussed in the introduction, the entropic cavity method predicts the existence of other thresholds, namely the 
dynamical and condensation thresholds. We have checked that the condensation one is the same for a coupled and individual ensemble (for $L\to +\infty$). This observation 
is consistent with the theorems of Appendix A. We also observe that the dynamical threshold of the coupled ensemble saturates towards 
the condensation one, for $K$ and $Q \geq 4$. For $K=3$ the dynamical and condensation thresholds coincide already for the individual ensemble. 
We consistently observe that they remain unchanged by coupling.
These observations deserve more investigations, in particular the nature of the condensed phase, the freezing of variables, 
the behavior of correlation functions and the possible relevance of the shape of clusters.  
 
The present work could find applications in a new method for proving lower bounds on $\alpha_s$ (and possibly $\alpha_c$). 
We hope that by choosing the right analyzable algorithms one may reach significant improvements
of the best existing lower bounds. One requirement on the algorithms is that they should 
be able to propagate in the bulk the ''seed`` given by the reduced hardness of the coupled instances at their 
boundaries. We have seen that this is the case for simple peeling-type  algorithms which are purely deterministic. 
Together with D. Achlioptas, we have observed that this is also the case for
classic stochastic algorithms if they are augmented by a suitable scheduling \cite{achlioptas-hassani-macris-urbanke}. 
Related ideas have been used recently within the context of a coupled CSP scheme for source coding \cite{aref-macris-urbanke-vuffray}. 
A  coupled low-density generator-matrix code is considered, and it is shown (numerically) that applying belief-propagation-guided-decimation with suitable scheduling, allows to approach the optimal rate-distortion curve of the individual code ensemble. 

\vskip 0.25cm

\appendix

\section{Free energy}\label{A}

We sketch the proof of the finite temperature analogs of theorems \ref{therm-limit-1} and 
\ref{therm-limit-2}.  An important consequence is that the free energies of the coupled and individual ensembles have 
the same singularities in the $(\alpha,\beta)$ plane (see \eqref{implication}). This is consistent with the fact that the average ground state energies,
and consequently the SAT-UNSAT thresholds,  are the same.

The  Gibbs distribution (at ``inverse temperature'' $\beta$) associated to the coupled CSP Hamiltonian \eqref{hamiltonian} is
\begin{equation}\label{measure}
 \mu_\beta(\underline x) = \frac{1}{Z_{\rm cou}} e^{-\beta\mathcal{H}_{\rm cou}(\underline x)},
\qquad 
Z_{\rm cou} = \sum_{\underline x} e^{-\beta\mathcal{H}_{\rm cou}(\underline x)},
\end{equation}
and the average free energy per node is
\begin{equation}\label{freecoupled}
 f_{N,L,w}(\alpha,\beta) = -\frac{1}{\beta LN}\mathbb{E}[\ln Z_{\rm cou}].
\end{equation}
The corresponding quantities  $\mathcal{H}_{\rm cou}^{\rm per}(\underline x)$ are associated a chain to with 
periodic boundary conditions (see section \ref{generalsetting}); these will be denoted by a superscript "per".
Note that to get these quantities for the underlying system, one sets $L=w=1$ in these definitions; the average free energy per node will be denoted by $f_{N}(\alpha,\beta)$.
\vskip 0.25cm
\noindent{\bf Remark about the entropy.}
The average entropy is defined as 
\begin{equation}
s_{N,L,w}(\alpha,\beta) = \frac{\partial}{\partial \beta^{-1}} f_{N,L,w}(\alpha,\beta) = \beta(\mathbb{E}[\langle \mathcal{H}_{\rm cou}\rangle] - f_{N,L,w}(\alpha,\beta))
\end{equation}
where $\langle - \rangle$ is the average with respect to \eqref{measure}. Theorems \ref{open-coup-comparison} and \ref{free-therm-limit} for the free energy have analogs for the average internal energy, and as a consequence also for the average entropy. Thus the entropy of coupled and individual ensembles have the same singularities in the $(\alpha,\beta)$ plane. 
This is consistent with the observation that the condensation threshold at zero temperature is the same for both ensembles.
\vskip 0.25cm
\begin{Theorem}[Comparison of open and periodic chains]\label{open-coup-comparison}
For general coupled CSP $[N,K,\alpha,w,L]$ ensembles we have
\begin{equation}
 \vert f_{N, L,w}(\alpha, \beta) - f_{N,L,w}^{\rm per}(\alpha,\beta)\vert \leq \frac{\alpha w}{L}
 \,.
\end{equation}
\end{Theorem}
\begin{proof}
We write (with the same notations than in the proof of theorem \ref{therm-limit-1})
\begin{equation}
Z_{\rm cou} = \sum_{\underline x} e^{-\beta(\mathcal{H}_{\rm cou}(\underline x)} = \sum_{\underline x^\prime, \underline x^{\prime\prime}} e^{-\beta\mathcal{H}_{\rm cou}^{\rm per}(\underline x^{\prime\prime})}
e^{-\beta(\mathcal{H}_{\rm cou}(\underline x^\prime, \underline x^{\prime\prime}) - \mathcal{H}_{\rm cou}^{\rm per}(\underline x^{\prime\prime}))}
\end{equation}
and from \eqref{2Mw} we deduce 
\begin{equation}
e^{-\beta Mw}Z_{\rm cou}^{\rm per} \leq Z_{\rm cou} \leq e^{\beta Mw} Z_{\rm cou}^{\rm per} 
\,.
\end{equation}
Applying $-\frac{1}{\beta NL}\log$ on each side of this inequality, we obtain the desired estimate.
\end{proof}

\begin{Theorem}[Comparison of individual and periodic ensembles]\label{free-therm-limit}
For $K$-SAT and $Q$-coloring the limits $\lim_{\rm N\to+\infty}f_{N,L,w}^{\rm per}(\alpha,\beta)$ and 
$\lim_{N\to+\infty} f_N(\alpha,\beta)$ exist, and are continuous in $(\alpha, \beta)$, for all $L$. Moreover,
\begin{equation}
\lim_{\rm N\to +\infty}f_{N,L,w}^{\rm per}(\alpha,\beta) = \lim_{N\to+\infty} f_N(\alpha,\beta)
 \,.
\end{equation}
\end{Theorem}
\vskip 0.25cm

Theorems \ref{open-coup-comparison} and \ref{free-therm-limit} yield (recall $\lim_{\rm therm}=\lim_{L\to +\infty}\lim_{N\to +\infty}$)
\begin{equation}\label{implication}
 \lim_{\rm therm} f_{N,L,w}(\alpha,\beta) = \lim_{\rm therm}f_{N,L,w}^{\rm per}(\alpha,\beta) = \lim_{N\to+\infty} f_N(\alpha,\beta).
\end{equation}

\begin{proof}
The proof of existence and continuity of limits for $N\to +\infty$ ($L$ fixed) is identical to \cite{BGP09}, so we do not repeat it here. 
The proof of the equality uses the same interpolating $r$-ensembles between the connected, ring and disconnected ensembles defined in subsection
\ref{interpolations}. The associated Gibbs measures, free energies and expectations will be denoted by scripts $r$, ${\rm conn}$, ${\rm ring}$ and ${\rm disc}$. 

By an argument similar to that of theorem \ref{open-coup-comparison} we have the analogs of 
\eqref{first}, \eqref{second}, \eqref{third},
\begin{equation}\nonumber
\begin{cases}
 -\lim_{N\to+\infty} \frac{1}{\beta LN}\mathbb{E}_{\rm conn}[\log Z_{\rm conn}] = \lim_{N\to+\infty} f_N(\alpha,\beta),\,\,{\rm for}\,\, L\,\, {\rm fixed},
\\
\\
 - \frac{1}{\beta LN}\mathbb{E}_{\rm disc}[\log Z_{\rm disc}] = f_N(\alpha,\beta) + O(N^{-1/2}), \,\,{\rm uniformly~in~} L,
\\
\\
 - \frac{1}{\beta LN}\mathbb{E}_{\rm ring}[\log Z_{\rm ring}] = f_{N,L,w}^{\rm per}(\alpha,\beta) + O(N^{-1/2}), \,\,{\rm uniformly~in~} L.
\end{cases}
\end{equation}
Thus, it is sufficient to show that 
\begin{equation}
 - \frac{1}{LN}\mathbb{E}_{\rm conn}[\log Z_{\rm conn}] 
\leq
 - \frac{1}{LN}\mathbb{E}_{\rm ring}[\log Z_{\rm ring}] 
\leq
 - \frac{1}{LN}\mathbb{E}_{\rm disc}[\log Z_{\rm disc}].
\label{double-inequ}
\end{equation}
To prove these inequalities we will use the $r$-ensembles. It suffices to check the analogs of \eqref{claim4} and \eqref{claim5}, namely for
an intermediate graph $\tilde G$,
\begin{equation}
 -\bigl(\mathbb{E}_r[\log Z_{G_r}\mid \tilde G] - \log Z_{\tilde G}\bigr) \leq -\bigl(\mathbb{E}_{r-1}[\log Z_{G_{r-1}}\mid \tilde G] - \log Z_{\tilde G}\bigr), 
\label{identical}
\end{equation}
and then average over $\tilde G$. 

Consider the graph $\tilde G$ and add a new constraint node $n$ to it. The precise way in which $n$ is connected
to the variable nodes is deferred to a later stage of the argument. We have
\begin{equation}
 \frac{Z_{\tilde G\cup n}}{Z_{\tilde G}} = e^{-\beta}\sum_{{\underline x}:\,  n {\rm ~is~UNSAT}}\mu_{\beta, \tilde G}(\underline x) + 
\sum_{{\underline x}: \, n {\rm ~is~SAT}}\mu_{\beta, \tilde G}(\underline x).
\end{equation}
This is equivalent to 
\begin{equation}
 \frac{Z_{\tilde G\cup n}}{Z_{\tilde G}} = 
1- (1-e^{-\beta})\sum_{{\underline x}:\,  n {\rm ~is~UNSAT}}\mu_{\beta, \tilde G}(\underline x).
\end{equation}
Taking the log and expectation over $n$ for a given $\tilde G$, we obtain
\begin{equation}
-\mathbb{E}\bigl[\log\frac{Z_{\tilde G\cup n}}{Z_{\tilde G}}\mid\tilde G\bigr] = 
-\mathbb{E}\bigl[\log\bigl\{1- (1-e^{-\beta})\sum_{{\underline x}:\,  n {\rm ~is~UNSAT}}\mu_{\beta, \tilde G}(\underline x)\bigr\}\mid\tilde G\bigr].
\label{log}
\end{equation}
Note that the left hand side is identical to that of \eqref{identical}. To compute the expectation we expand
$-\log(1-x) = \sum_{l\geq 1} \frac{x^l}{l}$,
\begin{align}
-\mathbb{E}\bigl[\log\frac{Z_{\tilde G\cup n}}{Z_{\tilde G}}\mid\tilde G\bigr] 
& = 
\sum_{l\geq 1} \frac{(1-e^{-\beta})^l}{l}
\nonumber \\ &
\times
\mathbb{E}\bigl[\sum_{{\underline x^{(1)},\ldots,\underline x^{(l)}}:\,  n {\rm ~is~UNSAT}}
\mu_{\beta, \tilde G}(\underline x^{(1)})
\ldots
\mu_{\beta, \tilde G}(\underline x^{(l)})
\bigr\}\mid\tilde G\bigr].
\label{exp}
\end{align}
The sum over ``real replicas'' ${\underline x^{(1)},\ldots,x^{(l)}}$ is over assignments such that $n$ is UNSAT for all $l$ of them,
so the expectation in \eqref{exp} equals
\begin{equation}
 \frac{1}{Z_{\tilde G}^l}\sum_{\underline x^{(1)}, \ldots, \underline x^{(l)}} 
e^{-\beta \sum_{\rho=1}^{l} \mathcal{H}_{\tilde G}(\underline x^{(\rho)})} 
\mathbb{E}\bigl[ \mathbbm{1}\{\text{$n$ UNSAT on all $\underline x^{(\rho)}$, $h=1,\ldots, l$}\}\mid \tilde G\bigr].
\label{abovesum}
\end{equation}

Up to this stage the arguments are completely general: they apply both to coloring and satisfiability.
We specialize the rest of the proof to $K$-SAT and leave coloring as an exercise.

We first derive \eqref{identical} for the $r$-ensemble that interpolates between the {\it connected and
ring} ensembles. This then implies the left inequality in \eqref{double-inequ}. 
Given $\tilde G$ and given a term ${\underline x^{(1)},\ldots,x^{(l)}}$ in \eqref{abovesum}, 
let $\mathcal{F}$ be the set of variable nodes with frozen bits, i.e those variable nodes such that the bit takes the same value in all assignments
$\underline x^{(1)}$ through $\underline x^{(l)}$. Below we  will also need the sets $\mathcal{F}_z = \mathcal{F}\cap V_z$. When $n$ is 
connected u.a.r. to the $LN$ variable nodes we go from $\tilde G$ to a $G_r$ graph  and 
\begin{equation}
 \mathbb{E}_r\bigl[ \mathbbm{1}\{\text{$n$ UNSAT on all $\underline x^{(\rho)}$, $h=1,\ldots, l$}\}\mid \tilde G\bigr] 
= \frac{1}{2^K}\biggl(\frac{\vert\mathcal{F}\vert}{LN}\biggr)^K.
\label{C1}
\end{equation}
On other hand when $n$ is first affected u.a.r. to a position $z$ and then connected u.a.r. to the $wN$ variables in $\cup_{k=0}^{w-1}
V_{z+k\mod L}$, we go from 
 $\tilde G$ to a $G_{r-1}$ graph and 
\begin{align}
 \mathbb{E}_{r-1} & \bigl[ \mathbbm{1}\{\text{$n$ UNSAT  on all $\underline x^{(\rho)}$, $h=1,\ldots, l$}\}\mid \tilde G\bigr] 
\nonumber \\ &
= \frac{1}{L}\sum_{z=-\frac{L}{2}+1}^{\frac{L}{2}}\frac{1}{2^K}\biggl(\frac{1}{wN}\sum_{k=0}^{w-1}\vert\mathcal{F}_{z+k\mod L}\vert\biggr)^K.
\label{C2}
\end{align}
By convexity, the quantity in \eqref{C1} is smaller than the one in \eqref{C2}. 
Using this fact together with \eqref{log}, \eqref{exp}, \eqref{abovesum}, we obtain the final inequality \eqref{identical}. 
This implies the left inequality in \eqref{double-inequ}.

The derivation of \eqref{identical} for the $r$-ensemble that interpolates between the {\it ring and disconnected} ensembles is similar. 
When $n$ is first affected u.a.r. to a position $z$, and then connected u.a.r. to $N$ variable nodes in the set $V_z$
we go from $\tilde G$ to a $G_{r-1}$ graph. Thus, 
\begin{align}
 \mathbb{E}_{r-1} & \bigl[ \mathbbm{1}\{\text{$n$ UNSAT  on all $\underline x^{(\rho)}$, $h=1,\ldots, l$}\}\mid \tilde G\bigr] 
\nonumber \\ &
= \frac{1}{L}\sum_{z=-\frac{L}{2}+1}^{\frac{L}{2}}\frac{1}{2^K}\biggl(\frac{\vert\mathcal{F}_{z}\vert}{N}\biggr)^K.
\label{C3}
\end{align}
Finally we notice that by convexity, $\eqref{C2}$ is smaller than $\eqref{C3}$, so that using 
again \eqref{log}, \eqref{exp} and \eqref{abovesum} we obtain the final inequality \eqref{identical}.
This now implies the right inequality in \eqref{double-inequ}.

\end{proof}

We now turn to the the case of XORSAT which has to be treated somewhat differently.
All definitions of average ground state energies and free energies are the same as usual.

\begin{Theorem}[Energy comparisons for XORSAT]\label{free-xorsat-limit}
For $K$-XORSAT with even $K$ the limits $\lim_{\rm N\to+\infty}f_{N,L,w}^{\rm per}(\alpha,\beta)$ and 
$\lim_{N\to+\infty} f_N(\alpha,\beta)$ exist, are continuous in $(\alpha, \beta)$, and are equal.
The same holds for the zero temperature quantities, i.e for $\lim_{\rm N\to+\infty}e_{N,L,w}^{\rm per}(\alpha)$ and 
$\lim_{N\to+\infty} e_N(\alpha)$.
\end{Theorem}

\begin{proof}
Existence and continuity of the limits for $K$-XORSAT follows from sub-additivity which was already proven in 
in \cite{franz-leone} for $K$ even\footnote{The argument in \cite{franz-leone} covers also $K$-SAT for even $K$, but a small modification of it 
extends the proof to odd $K$; however for XORSAT with odd $K$ it not clear how to extend the proof. 
Strangely enough, another case were such arguments break down is that of pure ferromagnetic
diluted interactions.}. Here we concentrate on the equality of limits.
The proof uses exactly the same interpolations as in the proof of theorem \ref{free-therm-limit}. 

First we prove the same relation 
 as in \eqref{double-inequ} for the case of XORSAT. For this we proceed exactly as in equs. \eqref{identical}-\eqref{abovesum} and reduce the problem to estimating 
the expectation 
\begin{equation}
 \mathbb{E}\bigl[\mathbbm{1}\{\text{$n$ UNSAT on all $\underline x^{(\rho)}$, $\rho=1,\ldots, l$}\}\mid \tilde G\bigr],
\label{exp-xor}
\end{equation}
according to the various ways of connecting the new constraint node $n$. It will be useful to represent
the indicator function in an algebraic way\footnote{The method used here can be used also for satisfiability and coloring and although it is 
somewhat longer, it may be useful when it is not obvious how to define ``frozen variables''.}. Suppose 
that $n$ connects to variable nodes $n_1,\cdots, n_K$, then
\begin{align}
 \mathbbm{1}&\{\text{$n$ UNSAT on all $\underline x^{(\rho)}$, $\rho=1,\ldots, l$}\} = 
\prod_{\rho=1}^l \frac{1}{2}\bigl(1-b_n\prod_{v=1}^K (-1)^{x_{n_v}^{(\rho)}}\bigr)
\nonumber \\ &
=\frac{1}{2^l}\sum_{0\leq r\leq l} (-1)^r b_n^r \sum_{\{\rho_1,\ldots,\rho_r\}\subset\{1,\ldots,l\}}
\prod_{v=1}^K\biggl\{(-1)^{x_{n_v}^{(\rho_1)}}\ldots (-1)^{x_{n_v}^{(\rho_r)}}\biggr\}.
\end{align}
When we take the expectation over $b_n\sim {\rm Bernoulli}(1/2)$ only the terms with $r$ even remain, 
\begin{equation}
 \frac{1}{2^l}\sum_{0\leq r\leq l}^{r{\rm ~even}}  \sum_{\{\rho_1,\ldots,\rho_r\}\subset\{1,\ldots,l\}}
\prod_{v=1}^K\biggl\{(-1)^{x_{n_v}^{(\rho_1)}}\ldots (-1)^{x_{n_v}^{(\rho_{r})}}\biggr\}.
\end{equation}
Now, it remains to compute the rest of the expectation on possible ways of connecting $n$. We define ``local overlap parameters''
\begin{equation}
 Q^{(\rho_1,\ldots,\rho_l)}_z = \frac{1}{N}\sum_{i=1}^N (-1)^{x_{iz}^{(\rho_1)}}\ldots (-1)^{x_{iz}^{(\rho_{r})}}.
\end{equation}

Let us first consider the interpolation
between the {\it ring and fully connected} ensembles. To go from $\tilde G$ to $G_r$ we connect $n$ u.a.r. among all 
$LN$ variable nodes $v=(i,z)$. Thus \eqref{exp-xor} becomes
\begin{equation}
 \frac{1}{2^l}\sum_{0\leq r\leq l}^{r{\rm ~even}}  \sum_{\{\rho_1,\ldots,\rho_r\}\subset\{1,\ldots,l\}}
\biggl\{\frac{1}{L}\sum_{z=-\frac{L}{2}+1}^{\frac{L}{2}}Q^{(\rho_1,\ldots,\rho_l)}_z\biggr\}^K.
\label{129}
\end{equation}
On the other hand to go from $\tilde G$ to $G_{r-1}$, we first affect $n$ to a position $z$ u.a.r. and connect its $K$ edges to variable nodes 
$v=(i,z+k)$ with $k\in \{0,\dots w-1\}$ u.a.r.. This time \eqref{exp-xor} becomes
\begin{equation}
 \frac{1}{2^l}\sum_{0\leq r\leq l}^{r{\rm ~even}}  \sum_{\{\rho_1,\ldots,\rho_r\}\subset\{1,\ldots,l\}}
\frac{1}{L}\sum_{z=-\frac{L}{2}+1}^{\frac{L}{2}}\biggl\{\frac{1}{w}\sum_{k=0}^{w-1}
Q^{(\rho_1,\ldots,\rho_l)}_{z+k\mod L}\biggr\}^K.
\label{130}
\end{equation} 
Convexity of the function $x^K$ for even $K$ implies that \eqref{129}$\leq$\eqref{130}, which then implies the left inequality
in \eqref{double-inequ}. Unfortunately at this point the argument breaks down for odd $K$ because we do not control the sign of the overlap parameters.

Consider now the interpolation between the {\it ring and disconnected} ensembles. When the extra node is first affected to $z$ u.a.r. and its 
$K$ edges connected u.a.r. to the $N$ variable nodes at the same position, we obtain for \eqref{exp-xor}
\begin{equation}
 \frac{1}{2^l}\sum_{0\leq r\leq l}^{r{\rm ~even}}  \sum_{\{\rho_1,\ldots,\rho_r\}\subset\{1,\ldots,l\}}
\frac{1}{L}\sum_{z=-\frac{L}{2}+1}^{\frac{L}{2}}\biggl\{
Q^{(\rho_1,\ldots,\rho_l)}_{z}\biggr\}^K.
\label{131}
\end{equation} 
Again convexity of $x^K$ for $K$ even implies that \eqref{130}$\leq$\eqref{131}, which then 
implies the right inequality in \eqref{double-inequ}. 

We have proven \eqref{double-inequ} for any finite $\beta$, and since $L$ and $N$ are finite, there is no difficulty
in taking the $\beta\to +\infty$ limit. This yields the zero temperature version of this inequality, namely \eqref{two-bounds}
applied to XORSAT.

Finally with the help of \eqref{two-bounds} and \eqref{double-inequ} we conclude (proceeding as in the previous proofs) that the average ground state and 
free energies of the individual and periodic ensemble are equal in the limit $N\to +\infty$, with $L$ and $w$ fixed. 

\end{proof}

\section{Review of the cavity method and survey propagation equations}\label{B}

The main assumptions of the cavity method draw on the concept of pure (or extremal or ergodic)
state. While this concept can be given a rigorous meaning for ``simple''  models \cite{Ruelle}, \cite{georgi}, it still forms 
a heuristic framework in the context of disordered spin systems. We refer the interested
reader to  \cite{Mezard-Montanari-book}, \cite{mezard-virasoro}, \cite{Talagrand}, \cite{stein-newman}, \cite{aizenman-wehr}
for more information and various approaches. 

Infinite volume Gibbs measures form a convex set whose extremal points play a special role and are called {\it pure states}. 
A crucial property of a pure state is that the correlations decay (usually exponentially fast) with
the graph distance. This is not true for non-trivial convex superpositions of pure states. 
For ``simple'' Ising-type models the number of pure states is ``small'' and they are related by a broken symmetry.
Disordered spin systems can have an exponential (in system size) 
number of pure states and the broken symmetry, if only there exist one, is hard to identify\footnote{Within the replica formalism 
it is a formal symmetry between "a number" of copies of the system.}. 
The growth rate of the number of 
pure states, is called the complexity. This is a notion analogous to the Boltzmann entropy, but at the level of pure states, instead of microscopic configurations,
for which one develops a new "level" of statistical mechanics.

We assume that this picture can be taken over to CSP and even coupled-CSP. Let $p$ index the set of pure states (we called them pure Bethe states in section
\ref{generalsetting}). 
The special feature about systems on random graphs is that they are  locally tree-like with high probability. 
Thus, since {\it for each pure state $p$} the correlations decay sufficiently fast, the marginals {\it of the pure state $p$} can be computed 
from the sum-product (or Belief Propagation) equations
\begin{align}
 \hat{\nu}_{cz \to iu}^{(p)}(x_{iu}) & \cong
\sum_{{x}_{\partial (cz) \backslash iu}} \psi_{cz}({x}_{\partial (cz)}) 
\prod _{jv \in \partial (cz) \backslash iu} \nu_{jv \to cz}^{(p)}(x_{jv}),
\label{proport1}
\\
\nu_{iu \to cz}^{(p)}(x_{iu}) & \cong \prod_{bv \in \partial (iu) \backslash cz} \hat{\nu}_{bv \to iu}^{(p)}(x_{iu}).
\label{proport2}
\end{align}
In \eqref{proport1}, \eqref{proport2} $\cong$ means that the right hand side has to be divided by a normalization factor to get a true
marginal on the left.
The free energy of the pure state $p$ is given by the Bethe expression,
\begin{align}
 \beta F^{(p)} & =  \sum_{cz}\ln\biggl\{\sum_{x_{\partial (cz)}}\psi_{cz}(x_{\partial (cz)}) 
\prod_{iu\in \partial (cz)}\nu_{iu \to cz}^{(p)} (x_{iu})\biggr\} 
\nonumber \\ &
      + \sum_{iu} \ln \biggl\{\sum_{x_{iu}} \prod_{cz\in\partial (iu)} 
           \hat{\nu}_{cz \to iu}^{(p)}(x_{iu})\biggr\}  
 \nonumber \\ &
       - \sum_{\langle cz,iu\rangle\in E} \ln\biggl\{\sum_{x_{iu}} \nu_{iu \to cz}^{(p)}(x_{iu})
\hat{\nu}_{cz \to iu}^{(p)}(x_{iu})\biggr\}.
\end{align}

To investigate the zero temperature limit $\beta\to+\infty$ we set
\begin{equation}\label{finitebeta}
 \nu_{iu\to cz}^{(p)}(x_{iu}) = 
 \frac{ e^{-\beta E_{iu\to cz}^{(p)}(x_{iu})}}{\sum_{x_{iu}\in\mathcal{X}}e^{-\beta E_{iu\to cz}^{(p)}(x_{iu})}},
\qquad \hat\nu_{cz\to iu}^{(p)}(x_{iu}) 
=\frac{ e^{-\beta \hat E_{cz\to iu}^{(p)}(x_{iu})}}{\sum_{x_{iu}\in\mathcal{X}}e^{-\beta \hat E_{cz\to iu}^{(p)}(x_{iu})}}.
\end{equation}
When $\beta\to \infty$, the sum-product equations \eqref{proport1} and \eqref{proport2} reduce to the min-sum equations
\begin{align}
E_{iu\to cz} (x_{iu})& = \min\bigl\{1, \sum_{bv \in \partial (iu) \backslash cz} 
\hat E_{bv \to iu}(x_{iu}) - C_{iu\to cz}\bigr\}
\nonumber
\\
&
\equiv  \mathcal{G}_{iu\to cz}\bigl[\{\hat E_{bv \to iu}\}_{bv\in \partial (iu) \backslash cz}\bigr],
\label{warn1}
\end{align}
\begin{align}
\hat E_{cz \to iu}(x_{iu}) & = 
\min_{{x}_{\partial cz \backslash iu}}\bigl\{ (1-\psi_{cz}({x}_{\partial (cz)}))+ \sum_{jv \in \partial (cz) \backslash iu} 
E_{jv \to cz}(x_j)\bigr\} - \hat C_{cz\to iu}
\nonumber \\ &
\equiv
\hat{\mathcal{G}}_{cz\to iu}\bigl[\{E_{jv \to cz}\}_{ jv \in \partial (cz) \backslash iu}\bigr].
\label{warn2}
\end{align}
Here, $C_{iu\to cz}$ and $\hat C_{cz\to iu}$ are normalization 
constants fixed so that \linebreak $\min_{x_{iu}} E_{iu\to cz}(x_{iu}) = \min_{x_{iu}} E_{cz\to iu}(x_{iu}) =0$.
The Bethe
formula for the free energy of a pure state reduces to an expression for its ground-state energy
\begin{equation}
 \lim_{\beta\to +\infty} \beta F^{(p)} = \mathcal{E} [\{E_{iu\to cz}^{(p)}(.), E_{cz\to iu}^{(p)}(.)\}], 
\end{equation}
where the functional $\mathcal{E}$ is given by 
\begin{align}
\mathcal{E} & [\{E_{iu\to cz}, E_{cz\to iu}\}] = \sum_{cz}\min_{x_{\partial (cz)}}
\bigl\{(1-\psi_{cz}(x_{\partial (cz)})) + \sum_{iu\in \partial (cz)}E_{iu \to cz} (x_{iu})\bigr\} 
\nonumber \\ &
+ \sum_{iu} \min_{x_{iu}}\bigl\{\sum_{cz\in\partial iu} 
 \hat{E}_{cz \to iu}(x_{iu})\bigr\}  
- \sum_{\langle cz,iu\rangle} \min_{x_{iu}} \bigl\{E_{iu \to cz}(x_{iu})   
+\hat{E}_{cz \to iu}(x_{iu})\bigr\}
\nonumber
\\ &
\equiv
\sum_{cz}
\mathcal{E}_{cz}\bigl[\{E_{iu \to cz}\}_{iu\in \partial cz}\bigr]  + \sum_{iu}\mathcal{E}_{iu}\bigl[\{\hat{E}_{cz \to iu}\}_{cz\in \partial iu}\bigr] 
\nonumber 
\\ &
- \sum_{\langle cz,iu\rangle}\mathcal{E}_{cz,iu}\bigl[\{E_{iu \to cz}, \hat{E}_{cz \to iu}\}\bigr].
\label{bethe}
\end{align}
We assume that the heuristic low temperature picture carries over to the zero temperature case. In this context pure states become 
clusters (in Hamming space) of minimizers of the Hamiltonian. Each cluster is characterized by a set of messages $\{E_{iu\to cz}^{(p)}(.), E_{cz\to iu}^{(p)}(.)\}$.
At zero temperature, two minimizers belonging to the same cluster can be connected by successive flips with infinitesimal energy cost, while for two minimizers 
belonging to different clusters this is not possible.

Now we wish to compute the complexity \eqref{complexity-definition} which counts the number of clusters. For this we introduce a generating 
function 
\begin{equation}
 \Xi(y) = \sum_p e^{-y \mathcal{E} [\{E_{iu\to cz}^{(p)}(.), E_{cz\to iu}^{(p)}(.)\}]} \,.
\label{parisiXi}
\end{equation}
When $y\to+\infty$ the sum is dominated by solutions of the min-sum equations with minimal Bethe energy. 
This object can be viewed as a partition function for the effective Hamiltonian \eqref{bethe} at 
inverse ``temperature'' $y$ (the so-called Parisi parameter). Now, if we take $\alpha$ in the SAT phase the minimum Bethe energy vanishes and the 
 complexity \eqref{complexity-definition} is given by
\begin{equation}
 \Sigma_{L,w}(\alpha) = \lim_{y\to +\infty} \lim_{N\to +\infty}\frac{1}{NL}\ln \Xi(y).
\label{comp-freeXi}
\end{equation}
A negative complexity signals that there are no zero energy states and that the system is in an UNSAT phase. When this happens one has to generalize
these formulas to allow for an energy dependent complexity (obtained by the Legendre transform of 
$\ln\Xi(y)$) but  this aspect will not concern us here.
For CSP's it can be shown that the min-sum messages take 
discrete values in a finite alphabet. Therefore we have
\begin{align}
 \Xi(y) = &\sum_{\{E_{iu\to cz}, \hat E_{cz\to iu}\}}
\biggl\{
\prod_{\langle iu,cz\rangle} e^{+y\mathcal{E}_{cz,iu}}
\biggr\}
\prod_{iu} \biggl\{ 
e^{-y\mathcal{E}_{iu}}\prod_{cz\in \partial (iu)}
\mathbbm{1}\bigl( E_{iu \to cz} = \mathcal{G}_{iu\to cz}\bigr)
\biggr\}
\nonumber \\ &
\times \prod_{cz} \biggl\{ 
e^{-y\mathcal{E}_{cz}}\prod_{iu\in \partial (cz)}
\mathbbm{1}\bigl( \hat E_{cz \to iu} = \hat{\mathcal{G}}_{cz\to iu}
\bigr)
\biggr\}.
\end{align}
The arguments of the functionals $\mathcal{E}_{iu}[-]$, $\mathcal{E}_{cz}[-]$, $\mathcal{E}_{iu,cz}[-]$ 
and $\mathcal{G}_{iu\to cz}[-]$, $\hat{\mathcal{G}}_{cz\to iu}[-]$ are the messages 
$\{E_{iu\to cz}(.), \hat E_{cz\to iu}(.)\}$; they are not explicitly written to 
ease the notation. It can easily be  
seen that this is the partition function of a new graphical model which is still sparse. Edges $\langle (c,z), (i,u)\rangle$ now correspond to degree two ``constraint'' nodes, and 
nodes $(c,z)$ and $(i,u)$ now correspond to ``variable'' nodes. 
Therefore \eqref{comp-freeXi} can be computed from the Bethe approximation for this new
model. The underlying assumption here is that the new effective model has a unique ``pure state`` with fast decaying correlations. 
This is called the level-1 cavity method. 
If this assumption breaks down, one should repeat the whole scheme, obtaining a level-2 cavity method (and so on). At level-1, the Bethe approximation
can be expressed in terms of new beliefs - called {\it surveys} - $Q_{iu\to cz}(E_{iu\to cz}(.))$ and 
$\hat Q_{cz\to iu}(\hat E_{cz\to iu}(.))$ that count 
the {\it fraction of clusters} $p$ for which $E_{iu\to cz}^{(p)}(.) = E_{iu\to cz}(.)$ and $E_{cz\to iu}^{(p)}(.) = E_{cz\to iu}(.)$. Note that these are the messages on 
the induced graph obtained by eliminating the degree two constraint nodes of the new model.
We have
\begin{align}
\ln \Xi(y) =  & \sum_{cz} \ln\biggl\{\sum_{\{E_{iu\to cz}\}_{iu\in \partial (cz)}} e^{-y \mathcal{E}_{cz}}\prod_{iu\in \partial cz} 
Q_{iu\to cz}\biggr\}
\nonumber \\
&
+\sum_{iu} \ln\biggl\{\sum_{\{\hat E_{cz\to iu}\}_{cz\in \partial (iu)}} e^{-y \mathcal{E}_{iu}}\prod_{cz\in \partial iu} Q_{cz\to iu}\biggr\} 
\nonumber \\
&
- \sum_{cz, iu}\ln\biggl\{\sum_{E_{iu\to cz},\hat E_{cz\to iu} } e^{-y \mathcal{E}_{iu,cz}} Q_{iu\to cz} \hat Q_{cz\to iu}\biggr\}.
\label{final}
\end{align}
The messages satisfy the {\it survey propagation equations}
\begin{align}
 Q_{iu\to cz}(E_{iu\to cz}) \cong  \sum_{\{\hat E_{bv\to iu}\}_{cz\in \partial (iu)}} & 
\mathbbm{1}\bigl(E_{iu\to cz} = \mathcal{G}_{iu\to cz}\bigr) e^{-y C_{iu\to cz}}
\nonumber \\ &
\times \prod_{bv\in \partial (iu)\backslash cz} 
Q_{bv\to iu}(\hat E_{bv\to iu}),
\label{sp1}
\end{align}
\begin{align}
 \hat Q_{cz\to iu}(\hat E_{cz\to iu}) \cong  \sum_{\{\hat E_{jv\to cz}\}_{jv\in \partial (cz)}} &
\mathbbm{1}\bigl(\hat E_{cz\to iu} = \hat{\mathcal{G}}_{cz\to iu}\bigr) e^{-y \hat C_{cz\to iu}}
\nonumber \\ &
\times \prod_{jv\in \partial (cz)\backslash iu} 
Q_{jv\to cz}(E_{jv\to cz}),
\label{sp2}
\end{align}
where again $\cong$ means that the right hand side has to be normalized. 

In the SAT phase one takes $y\to +\infty$ in order to compute the complexity. This has the effect of 
reducing the sums in \eqref{sp1}, \eqref{sp2} and \eqref{final}, to surveys such that $C_{iu\to cz}=\hat C_{cz\to iu}=0$ and 
$\mathcal{E}_{cz}=\mathcal{E}_{iu}=\mathcal{E}_{iu, cz}=0$.

\vskip 1cm

\noindent{\bf \Large Acknowledgments.} The work of Hamed Hassani has been supported by Swiss National 
Science Foundation Grant no 200021-121903. N.M thanks Marc M\'ezard and Toshiyuki Tanaka for instructive discussions on coupling for the SK model.
We thank Dimitri Achlioptas for interesting discussions on algorithmic aspects.

\end{document}